\documentclass[a4paper,11pt,onecolumn,accepted=2022-06-02]{quantumarticle}
\pdfoutput=1
\usepackage{graphicx,epsfig}
\usepackage{amsmath}
\usepackage{amsfonts}
\usepackage{braket}
\usepackage{amsthm}
\usepackage{array}
\usepackage[dvipsnames]{xcolor}
\usepackage{algorithm} 
\usepackage{algpseudocode} 
\usepackage{overpic}
\usepackage{hyperref}
\usepackage{tabularx}
\usepackage{multirow}
\usepackage{array}
\usepackage[numbers]{natbib}
\newcolumntype{P}[1]{>{\centering\arraybackslash}p{#1}}

\newtheorem{result}{Result}

\newtheorem{theorem}{Theorem}
\newtheorem{lemma}[theorem]{Lemma}

\theoremstyle{definition}
\newtheorem{definition}[theorem]{Definition}

\DeclareMathOperator{\Tr}{Tr}
\DeclareMathOperator*{\argmin}{\arg\!\min}
\DeclareMathOperator*{\argmax}{\arg\!\max}
\DeclareMathOperator{\EX}{\mathbb{E}}

\newcommand{\Y}{\rule{0pt}{3.2ex}}

\newcommand{\mc}[1]{\mathcal{#1}}
\newcommand{\ms}[1]{\mathsf{#1}}

\newcommand{\mb}[1]{\mathbb{#1}}

\newcommand{\N}{\mathbb N}
\newcommand{\R}{\mathbb R}

\newcommand{\C}{\mathbb C}

\newcommand{\hil}{\mathcal{H}} 
\newcommand{\tr}[1]{\mathrm{Tr}\left(#1\right)} 
\def\<{\langle}
\def\>{\rangle}
\newcommand{\fii}{\varphi}

\newcommand{\dd}{\mathrm{d}}

\begin{document}


\title{ Multi-armed quantum bandits: Exploration versus exploitation when learning properties of quantum states}
\author{Josep Lumbreras}
\author{Erkka Haapasalo} 
\affiliation{Centre for Quantum  Technologies,  National University of Singapore, Singapore}
\author{Marco Tomamichel}
\affiliation{Centre for Quantum  Technologies,  National University of Singapore, Singapore}
\affiliation{Department of Electrical and Computer Engineering, Faculty of Engineering, National University of Singapore, Singapore}

\begin{abstract}
We initiate the study of tradeoffs between exploration and exploitation in online learning of properties of quantum states. Given sequential oracle access to an unknown quantum state, in each round, we are tasked to choose an observable from a set of actions aiming to maximize its expectation value on the state (the reward). Information gained about the unknown state from previous rounds can be used to gradually improve the choice of action, thus reducing the gap between the reward and the maximal reward attainable with the given action set (the regret). We provide various information-theoretic lower bounds on the cumulative regret that an optimal learner must incur, and show that it scales at least as the square root of the number of rounds played. We also investigate the dependence of the cumulative regret on the number of available actions and the dimension of the underlying space. Moreover, we exhibit strategies that are optimal for bandits with a finite number of arms and general mixed states. 
\\
\end{abstract}

\maketitle


\section{Introduction}

Multi-armed bandits have been studied for nearly a century and it is nowadays a very active research field, in terms of both theory and applications \cite{banditalgorithm,theory2,theory1,applications1}. The multi-armed bandit problem is a simple model of decision making with uncertainty that lies in the class of classical reinforcement learning problems. Given a set of arms, a learner interacts sequentially with these arms sampling a reward at each round and the objective of the learner is to identify the arm with largest expected reward while maximizing the total cumulative reward. The problem that the learner faces is a trade-off between exploration and exploitation: one wants to explore all the arms to identify the best one but also wants to exploit those arms that give the best rewards. Bandit algorithms are online, which means that the strategy is adaptive at each round and is learned from previously observed events. This is one of the reasons why nowadays bandit algorithms are applied to online services such as advertisement recommendation~\cite{advertisement} or dynamical pricing \cite{dynamical}.

The problem was introduced by Thompson in 1933 \cite{thompson33} where the original motivation was to study medical trials. The problem was to decide which treatment to use in the next patient given a set of different treatments for a certain disease. The mathematical formalization and popularization of the field is due to the mathematician and statistician Robbins~\cite{robbins52}, who described the problem as a set of arms such that the reward associated to each arm follows an independent unknown probability distribution. This model is usually referred to as multi-armed stochastic bandits.

The main technique to solve this problem is based on upper confidence bounds that were introduced by Lai and Robbins \cite{lairobbins85}. There are many different models of bandits, including adversarial models, but the main setting that we are interested for our work are the multi-armed stochastic linear bandits, first considered by Aurer in \cite{aurer02}. In this model, the arms can be viewed as vectors and the expected reward of each arm is given by the inner product of the vector associated to the arm and an unknown vector that is the same for all arms. Since the literature of bandits is very extensive we refer to the book by Lattimore and Szepesvári \cite{banditalgorithm} for a comprehensive  review of bandit algorithms.

Quantum algorithms for the classical multi-armed stochastic bandit problem have been proposed recently~\cite{quantumbandits,quantumbandits2}. A quantum version of the Hedging algorithm, which is related to the adversarial bandit model, has also been studied~\cite{hedging}. These algorithms investigate potential improvements on the respective classical bandit algorithms when a quantum learner is given superposition access to the oracle, i.e., it can probe rewards for several arms in superposition.

Our work treats multi-armed bandits from a different perspective, focusing instead on the learning theory for quantum states. In our model, which we call the \emph{multi-armed quantum bandit} model, the arms correspond to different observables or measurements and the corresponding reward is distributed according to Born's rule for these measurements on an unknown quantum state. We are interested in the optimal tradeoff between exploration (i.e.\ learning more about the unknown quantum state) and exploitation (i.e.\ using acquired information about the unknown state to choose the most rewarding but not necessarily most informative measurement). More precisely, we consider a learner that at each round has access to a copy of an unknown quantum state $\rho$ (the \emph{environment}) and has to choose an observable from a given set $\mathcal{A}$ (the \emph{action set}) in order to perform a measurement and receive a reward. The \emph{reward} is sampled from the probability distribution associated with the measurement of $\rho$ using the chosen observable. The figure of merit that we study is the \emph{cumulative expected regret}, the sum over all rounds of the difference between the maximal expected reward over $\mathcal{A}$ and the expected actual reward associated to the chosen observable at each round. If the actions set is comprised of all rank-1 projectors, we arrive at the problem of finding the maximal eigenvector of the unknown state. The detailed model is introduced in Section~\ref{sec:model}.

This notion of regret has a natural interpretation in some physical settings. Consider for example a sparse source of single photons with fixed but unknown polarisation (i.e., the reference frame is unknown). In order to learn the unknown reference frame we can perform a phase shift and apply a polarisation filter, adjusting the phase (the action) until the photons consistently pass the filter (the reward). The regret would then be proportional to the energy absorbed in the filter during our learning process. 
As mentioned above, classical bandit algorithms find applications in recommendation  systems,
and  it  is a  valid  question whether  quantum bandit  models could  be  used  similarly for
recommendation systems for quantum devices. However, we think that  our present model does
not  capture any interesting recommendation tasks, and that a further generalisation of our model
to allow for contextual information would be necessary. This is outside the scope of the current
work.

In this work we give lower and upper bounds on the minimal regret attainable by any strategy in terms of the number of rounds, $n$, the number of arms, $k$, and the dimension of the Hilbert space, $d$, for different sets of environments and actions. We discuss our results in detail in Section~\ref{sec:results} and give the proofs in Sections~\ref{lower},~\ref{sec:lowergeneral} and~\ref{sec:algorithms}.
For the upper bounds in Section~\ref{sec:algorithms} we use an algorithmic approach, i.e, we try to find an algorithm that has low regret for all states and perform the regret analysis. Here our main technique is to apply known classical algorithms using a vector representation of our problem (see Lemma \ref{lem:subgaussianrewards}). 
For the lower bounds in Sections~\ref{lower} and~\ref{sec:lowergeneral} we use information theoretical techniques, focusing on finding environments for which all strategies perform ``poorly''. 

It is worth stressing here that our model falls within the class of (classical) multi-armed stochastic linear bandits. This is not surprising: most problems in quantum tomography, metrology or learning can be converted to a classical problem with additional structure, since, in the end, we are interested in learning a matrix of complex numbers constituting the density operator. The question of interest then is whether the structure imposed by the quantum problem is helpful to find more efficient algorithms. Indeed, the fact that our model is a subset of a more general class of problems for which classical algorithms are known does not imply that these algorithms yield the smallest possible regret for the subset of problems we consider. In fact, our results show that the classical algorithms are optimal in most cases, although we identify at least one setting (learning pure states using the set of all rank-1 projectors) where we expect bespoke algorithms to perform better. Moreover, while inspired by lower bounds on the classical multi-armed stochastic bandit problem, our bounds require novel constructions that are specific to the quantum state space. On the one hand, due to correlations in the reward distribution associated to each arm (observable), the standard multi-armed stochastic bandit lower bounds proofs do not apply to our case. On the other hand, lower bounds for linear stochastic bandits have been studied only for specific action sets like the hypercube or unit sphere (see \cite{hypercube} and \cite{unitsphere}) and there is no combination of action sets and environments in the classical proofs that can be mapped to our particular class of problems. For that reason, known classical regret lower bounds do not apply to our case. 

Beyond bandits, there are other reinforcement learning frameworks where one can model more complex environments where actions have long-term consequences.
One example are Markov decision processes (MDPs) where the learner interacts with the environment by choosing an action, receiving a reward, and observing a state (or partial information). The reward received by the learner not only depends on the action, it also depends on the state and this state also evolves depending on the action taken by the learner.
MDPs have been generalized to the quantum setting (see \cite{barry11,ying2021} where the underlying states are quantum states and the evolution and rewards are generated following quantum processes. Our model fits as a restricted version of \cite{ying2021} since actions do not prompt state transformations in our model. The model presented in \cite{ying2021} also involves measurements providing additional partial information. However, our model, being more specific, allows us to reach more detailed results including lower and upper bounds for the cumulative regret.

Finally, we note that one can try to apply techniques from quantum state tomography~\cite{tomographyreview} or shadow tomography~\cite{shadow} in this setting using the observables of the action set in order to estimate the unknown quantum state. If we learn the unknown quantum state (or some approximation thereof) we can choose the best action in order to minimize the regret. However this strategy is not optimal: quantum state tomography algorithms can be thought of as pure exploration strategies since the algorithm only cares about choosing the action that helps us to learn most about the unknown quantum state, which is not necessarily the action that minimizes the regret. In particular, our setting is thus different from the setting of online learning of quantum states~\cite{onlinelearning}. There, the learner is tasked to produce an estimate $\omega_t$ of $\rho$ in each round and the regret is related to the quality of this estimate. The main difference with our model is that we do not require our policy to produce an estimate of the unknown quantum state, just to choose an observable with large expectation value on the unknown quantum state.


\section{Multi-armed quantum bandits}
\label{sec:model}

We assume familiarity with basic concepts and notations in probability theory and quantum information theory. Here we introduce notations needed to set up the model and present our results. Some more concepts will be introduced in the sections of proofs. We denote the natural logarithm as $\log$. We define $[n] := \{1, 2, \ldots, n\}$ for $n \in \mathbb{N}$.
We restrict our attention to finite-dimensional quantum systems. Let $\mathcal{S}_d = \lbrace \rho \in \mathbb{\mathbb{C}}^{d \times d}: \rho \geq 0 \wedge \Tr ( \rho ) = 1 \rbrace$ denote the set of positive semi-definite operators with unit trace, i.e \emph{quantum states} that act on a $d$-dimensional Hilbert space $\mathbb{C}^d$. Pure states are rank-1 projectors in the set of states, given by $\mathcal{S}_d^* = \{ \rho \in \mathcal{S}_d : \rho^2 = \rho \}$. Moreover, \emph{observables} are Hermitian operators acting on $\mathbb{C}^d$, collected in the set $\mathcal{O}_d = \lbrace O\in \mathbb{C}^{d\times d} : O^{\dagger} = O \rbrace $. An observable $O$ is called \emph{traceless} if $\Tr(O) = 0$ and it is called \emph{sub-normalised} if $\| O \| \leq 1$, where $\|\cdot\|$ denotes the operator norm. For two quantum states $\rho, \sigma \in \mathcal{S}_d$ we write their trace distance as $\frac12 \| \rho - \sigma \|_1$ where $\| X \|_1 = \Tr |X|$ is the Schatten 1-norm.

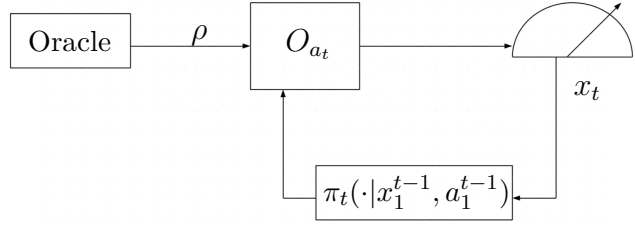
\begin{figure}
\centering
\begin{overpic}[percent,width=0.55\textwidth]{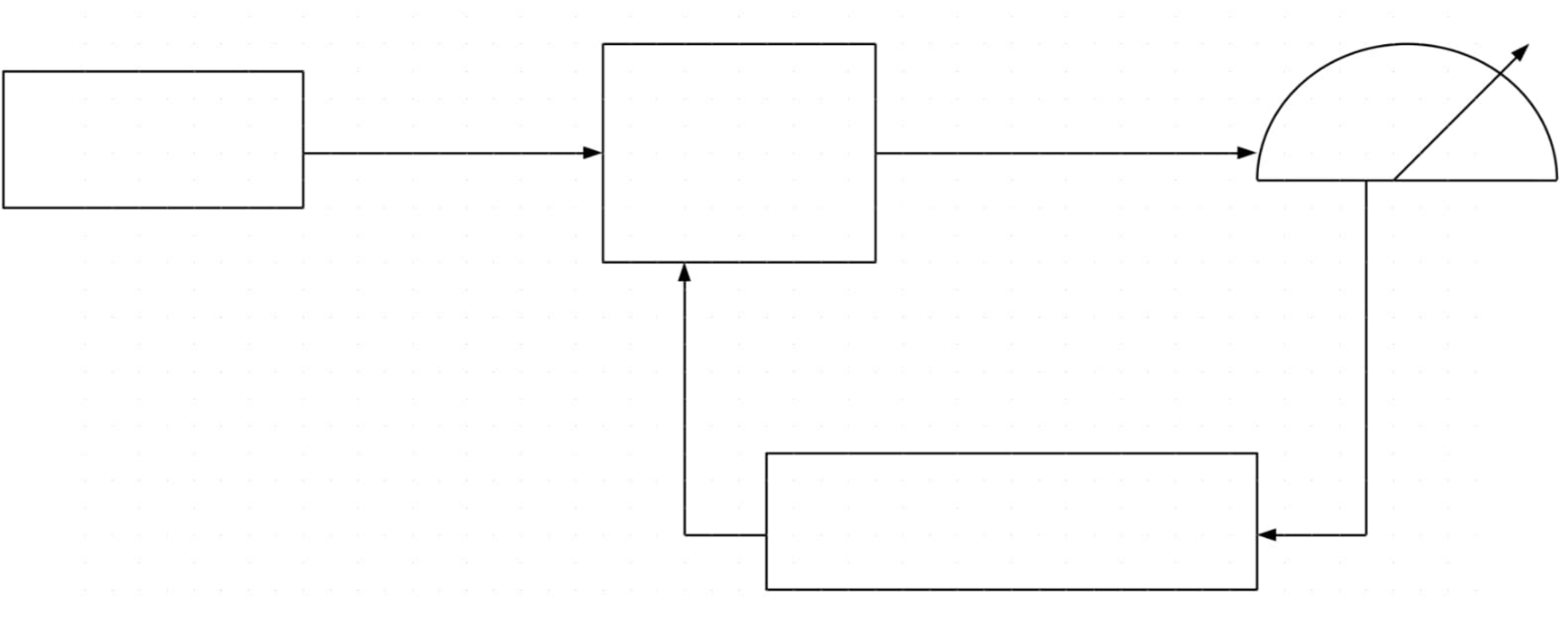}
\put(3,30){\color{black} Oracle}
\put(29,32){\large$\rho$}
\put(44,30){\color{black}\large$O_{a_t}$}
\put(50,6){$\pi_t (\cdot \vert x^{t-1}_1,a_1^{t-1} )$}
\put(90,23){\large$x_t$}
\end{overpic}
\caption{Scheme for the multi-armed quantum bandit setting.}
\label{fig:scheme}
\end{figure}

\bigskip

We will consider bandits with a finite set of actions, which we call discrete bandits, and bandits with a general, potentially continuous, set of actions.

\subsection{Discrete bandits}

\begin{definition}[Multi-armed quantum bandit]
	Let $d \in \mathbb{N}$. 
	A $d$-dimensional \emph{discrete multi-armed quantum bandit} is given by a finite set $\mathcal{A} \subseteq \mathcal{O}_d$ of observables that we call the \emph{action set}. The bandit is in an \emph{environment}, a quantum state $\rho \in \Gamma$, that is unknown but taken from a set of \emph{potential environments} $\Gamma \subseteq \mathcal{S}_d$. The bandit problem is fully characterized by the tuple $(\mathcal{A}, \Gamma)$.
\end{definition}
As potential environments in this work we consider either general states or only pure states. We will consider discrete bandits with (almost) arbitrary action sets and bandits with structured action sets, for example given by Pauli measurements. 

Let us now fix a bandit with unknown state $\rho$ and action set $\mc A = \{O_1, O_2, \ldots, O_k\}$ of cardinality $k = |\mathcal{A}|$. For each observable indexed by $a \in [k]$, we also introduce its spectral decomposition,
\begin{align}
	O_{a} = \sum_{i=1}^{d_a} \lambda_{a,i} \Pi_{a,i}, \label{eq:decomp}
\end{align}
where $\lambda_{a,i} \in \mathbb{R}$ denote the $d_a \leq d$ distinct eigenvalues of $O_a$ and $\Pi_{a,i}$ are the orthogonal projectors on the respective eigenspaces. We also introduce the spectrum $\Lambda_a = \{ \lambda_{a,1}, \ldots, \lambda_{a,d_a} \}$ of the observable $O_a$.

With this in hand we can now describe the learning process (see also Figure~\ref{fig:scheme}). In each round $t \in \mathbb{N}$, the learner (probabilistically or deterministically) chooses an action $A_t \in [k]$ and measures the observable $O_{A_t}$ on a copy of the unknown quantum state $\rho$. The result of the measurement is the reward $X_t \in \Lambda_{A_t}$. 
More specifically, given the spectral decompositions in~\eqref{eq:decomp}, the conditional probability distribution of the reward $X_t$ is given by Born's rule:
\begin{align}
	\Pr[X_t = x | A_t = a] = P_{\rho}(x | a) = \begin{cases} \Tr(\rho \Pi_{a,i} ) & \textrm{if } x = \lambda_{a,i} \\ 0 & \textrm{else} \end{cases} \label{eq:Born} \,.
\end{align}
for any $a \in [k]$. The probability is thus nonzero only if $x \in \Lambda_a$.
Note in particular that the conditional expectation of $X_t$ is given by $\EX_{\rho}[X_t | A_t] = \Tr ( \rho O_{A_t} )$. The learner's choice of action is determined by a policy.

\begin{definition}
	\label{def:policy}
	A \emph{policy} (or algorithm) for a multi-armed quantum bandit is a set of (conditional) probability distributions $\pi = \{ \pi_t \}_{t \in \mathbb{N}}$ on the action index set $[k]$ of the form
	\begin{align}
		\pi_t ( \cdot \vert a_1 , x_1 , ... , a_{t-1}, x_{t-1} ) ,
	\end{align}
	defined for all valid combinations of actions and rewards $(a_1, x_1), \ldots, (a_{t-1}, x_{t-1})$ up to time $t-1$.
\end{definition}

Then, if we run the policy $\pi$ on the state $\rho$ over $n \in \mathbb{N}$ rounds, we can define a joint probability distribution over the set of actions and rewards as
\begin{align}\label{probdens}
	P_{\rho ,\pi }(a_1,x_1,...,a_n,x_n) := \prod_{t=1}^n \pi_t ( a_t \vert a_1 , x_1 , ... , a_{t-1}, x_{t-1} ) P_{\rho}(x_t|a_t).
\end{align}
This distribution fully describes the random learning process.

\subsection{Regret}

The objective of the learner is to find the action (observable) that in expectation maximizes the reward. This is equivalent to minimizing the {expected cumulative regret}.

\begin{definition}
Given a multi-armed quantum bandit problem $(\mathcal{A}, \Gamma)$, a state $\rho \in \Gamma$, a policy $\pi$ and $n \in \mathbb{N}$, we define the \emph{expected cumulative regret} as
\begin{align}\label{regret}
    R_n (\mathcal{A}, \rho ,\pi ) := \sum_{t=1}^n \max_{O \in \mathcal{A}} \Tr ( \rho O ) - \EX_{\rho , \pi} [X_t] \,,
\end{align}
where the expectation value is taken with respect to the probability density~\eqref{probdens}.
\end{definition}
We note that $\EX_{\rho,\pi} [X_t] = \Tr( \rho O_{A_t})$. We define the \emph{sub-optimality gap} as,
\begin{align}\label{suboptimality}
\Delta_a := \max_{O \in \mathcal{A}} \Tr ( \rho O) - \Tr( \rho O_a),
\end{align}
for $a \in [k]$. The sub-optimality gap represents the relative loss between the optimal measurement and the measurement induced by $O_a \in \mathcal{A}$. Note that $\Delta_a \geq 0$ and the equality is only achieved by the optimal observables (it does not need to be unique). 
In order to analyse the expected cumulative regret it is often more convenient to work with the equivalent expression given by
\begin{align}\label{regretarms}
R_n(\mathcal{A}, \rho , \pi) &= \sum_{a=1}^k \sum_{t=1}^n \EX_{\rho,\pi}[(\max_{O \in \mathcal{A}} \Tr ( \rho O )-X_t) \mathbb{I}\lbrace A_t = a \rbrace ] \\
							  &=	 \sum_{a=1}^k \sum_{t=1}^n \Delta_a \EX_{\rho,\pi}[\mathbb{I}\lbrace A_t = a \rbrace]	\\			
							  &=	\sum_{a=1}^k \Delta_a \EX_{\rho,\pi}[T_a(n)] \,,
\end{align}
where $T_a(n) = \sum_{t=1}^n \mathbb{I}\lbrace A_t = a \rbrace$ is a random variable that tells us how many times we have picked the action~$a$ over $n$ rounds. 

For a fixed policy $\pi$ the cumulative expected regret depends on $\rho$. Thus, in order to determine how well a policy performs for a given set $\Gamma$ of environments we need to find a state-independent metrics.
\begin{definition}
	Given a multi-armed quantum bandit problem $(\mathcal{A}, \Gamma)$, a policy $\pi$ and $n \in \mathbb{N}$, we define the \emph{worst-case regret} as
	\begin{align}
		R_n (\mathcal{A}, \Gamma , \pi) = \sup_{\rho \in \Gamma} R_n ( \mathcal{A}, \rho, \pi).
	\end{align}
	Moreover, the \textit{minimax regret} is defined as, 
	\begin{align}
		R_n (\mathcal{A}, \Gamma) = \inf_{\pi} R_n ( \mathcal{A}, \Gamma , \pi) ,
	\end{align}
	where the infimum goes over all possible policies of the form in Definition~\ref{def:policy}.
\end{definition}

The minimax regret is a measure of how difficult the multi-armed quantum bandit problem is and we will attempt to find upper and lower bounds on $R_n (\mathcal{A}, \Gamma)$ as a function of $n$ and the Hilbert space dimension. A small value of $R_n (\mathcal{A}, \Gamma)$ means that the problem is less difficult to learn.

\subsection{General bandits}\label{subsec:generalbandits}

Often it is natural to allow a continuous set of possible observables, for example all rank-1 projectors. This leads us to a more general definition of multi-armed quantum bandits that works for any measurable space of actions. 

We will only mention the changes in the formalism in comparison to the discrete case, and build on the notions introduced above.

\begin{definition}
	Let $d \in \mathbb{N}$.
	A general $d$-dimensional \emph{multi-armed quantum bandit} is given by a measurable space $(\mc A,\Sigma)$, where $\Sigma$ is a $\sigma$-algebra of subsets of $\mc A$. 
	A \emph{policy} for such a bandit is given by conditional probability measures $\pi_t(\cdot|a_1,x_1,\ldots,a_{t-1},x_{t-1}) : \Sigma\to[0,1]$.
\end{definition}
The reward distribution $P_{\rho}(x | a)$ conditioned on playing arm $a\in\mc A$ when the environment state is $\rho$ is still given by Born's rule in~\eqref{eq:Born} and, importantly, remains discrete. To simplify our presentation we in the following also assume that the spectra of the observables satisfy $\Lambda_a \subset \mathcal{X}$ for some finite set $\mathcal{X}$, i.e., we only allow for a discrete set of possible rewards independently of the choice of action. This is for example trivially satisfied for the case of rank-1 projectors, where $\mathcal{X} = \{0,1\}$ is trivial. We also need to assume that the function $a \mapsto P_{\rho}(x | a) $ is $\Sigma$-measurable for all environment states $\rho$ and $x \in \mathcal{X}$. 

The environment state $\rho$ and the policy $\pi$ define the probability measure $P_{\rho,\pi}: (\Sigma \times \mathcal{X})^{\times n} \to[0,1]$ through
\begin{align}
	&P_{\rho,\pi}(A_1,x_1,\ldots, A_n, x_n)\nonumber\\
	=&  \int_{A_1} \cdots \int_{A_n} P_{\rho}(x_n|a_n) \pi_n(da_n|a_1,x_1,\ldots,a_{n-1},x_{n-1}) \cdots 
 	P_{\rho}(x_1|a_1) \pi_1(\mathrm{d} a_1) \label{eq:contdistr}
\end{align}

This map will be treated as a measure although it is not strictly speaking a set function. However, this should cause no confusion.

In order to present a formula for regret, we need the continuous counterpart for the function $a\mapsto\mathbb{E}_{\rho,\pi}\big(T_a(n)\big)$ of the discrete case. To this end, we define, for all $t=1,\ldots,\,n$, the margins $P_{\rho,\pi}^{(t)}:\Sigma\times\mc X\to[0,1]$ through
\begin{align}
P_{\rho,\pi}^{(t)}(A,x)=\sum_{x_1\in\mc X}\cdots\sum_{x_{t-1}\in\mc X}\sum_{x_{t+1}\in\mc X}\cdots\sum_{x_n\in\mc X}P_{\rho,\pi}(\mc A,x_1,\ldots,\mc A,x_{t-1},A,x,\mc A,x_{t+1},\ldots,\mc A,x_n).
\end{align}
We now define the measure $\gamma_{\rho,\pi}:\Sigma\to[0,n]$ through
\begin{align}\label{eq:gammameasure}
\gamma_{\rho,\pi}(A)=\sum_{t=1}^n\sum_{x\in\mc X}P_{\rho,\pi}^{(t)}(A,x).
\end{align}
We denote by $E_{\rho|a}$ the expectation value of the conditional distribution $x\mapsto P_\rho(x|a)$ and by $\Delta_a$ the 
supremum of the function $b\mapsto E_{\rho|b}-E_{\rho|a}$ for any arm $a\in\mc A$; this is naturally the continuous analogue of the sub-optimality gap. Note that, in all the cases we study, $\mc A$ is a compact topological space, $\Sigma$ is the associated Borel $\sigma$-algebra, and $a\mapsto E_{\rho|a}$ is continuous for every environment $\rho$, so this supremum (which is also a maximum) technically makes sense. Using the measure $\gamma_{\rho,\pi}$ and the above sub-optimality gap, we may define the expected cumulative regret analogously to the discrete case:
\begin{align}
R_n(\rho,\pi,\mc A)&:=\sum_{t=1}^n\mb E_{\rho,\pi}(\Delta_{a_t})\label{eq:contRegr2}\\
&=\sum_{t=1}^n\sum_{x_1\in\mc X}\cdots\sum_{x_n\in\mc X}\int_{\mc A^n}\Delta_{a_t}\,P_{\rho,\pi}(da_1,x_1,\ldots,da_n,x_n)\\
&=\int_{\mc A}\Delta_a\,d\gamma_{\rho,\pi}(a).\label{eq:contRegr1}
\end{align}
Note that Equation \eqref{eq:contRegr2} implies that, for the optimal strategy which always plays the best arm, the regret vanishes. The worst-case and minimax regrets are defined in the obvious way.


A generalisation of these definitions to environments comprised of states on infinite-dimensional Hilbert spaces is straight-forward, but outside the scope of this work. Especially, in the infinite-dimensional setting, it is natural that the reward set is also continuous.

\section{Main results and discussion}
\label{sec:results}

We present our main results of lower and upper bounds on the minimax regret for different environments and action sets. The results are summarized in Table \ref{table:summary}.

\begin{table}[h!]
\centering
\begin{tabular}{  |c|c|c|c| } 
\hline
 & \multicolumn{2}{c|}{Discrete} \vline & Continuous \\
\cline{2-3}
  & Arm-limited & Dimension-limited & (all rank-1 projectors) \\
\hline  
\Y
 \multirow{4}{*}{$\Gamma= S_d$} & $ R_n (\mathcal{A},\Gamma ) =\Omega (\sqrt{kn} ) $ $(k<d^2)$ & $ R_n (\mathcal{A},\Gamma ) =\Omega\left(d\sqrt{n} \right) $ & $ R_n (\mathcal{A},\Gamma ) = \Omega\left(\sqrt{n} \right) $\\
& Result \ref{res:pauliobservables} & Result \ref{res:pauliobservables} & Result \ref{res:generalcontinous} \\
\cline{2-4}
\Y
& \multirow{2}{*}{$R_n (\mathcal{A},\Gamma )= \widetilde{O}(\sqrt{kn})$ } & \multirow{2}{*}{ $R_n (\mathcal{A},\Gamma ) = \widetilde{O}\left( d\sqrt{n} \right) $} & $R_n (\mathcal{A},\Gamma ) = \widetilde{O}(d^2\sqrt{n}) $ \\
 & & & Result \ref{generalalgorithm} \\
\cline{1-1}
\cline{4-4}
\Y

\multirow{4}{*}{$\Gamma= S^*_d$} & Result \ref{res:UCB} & Result \ref{res:UCB} & $R_n (\mathcal{A},\Gamma ) = \widetilde{O}(\sqrt{n}), (d=2) $ \\
 &  & & Result \ref{res:purealgorithm} \\

\cline{2-4}
\Y
& \multicolumn{2}{c|}{$R_n (\mathcal{A},\Gamma )  = \Omega\left(\sqrt{n} \right) $ $(d=2,k=3)$} & $ R_n (\mathcal{A},\Gamma )  = ? $\\
& \multicolumn{2}{c|}{Result \ref{res:purepauli}}  &  \\
\hline

\end{tabular}
\caption{Scaling of the minimax regret in terms of the number of rounds, $n$, dimension of the Hilbert space, $d$, and number of actions, $k$, for discrete bandits. We study general state environments, $\mathcal{S}_d$, and pure-states environments, $\mathcal{S}^*_d$. For discrete actions sets we differentiate between arm-limited  (the number of arms is smaller than the degrees of freedom in the quantum state space) and dimension-limited (the number of arms can be arbitrary large). }
\label{table:summary}
\end{table}

\subsection{Lower bounds}

We start giving the lower bounds results, see Section \ref{lower} and Section \ref{sec:lowergeneral} for the detailed statements and proofs. We first give a generic lower bound that works for (almost) all discrete models.

\begin{result}\label{res:generaldiscrete}
Consider a discrete multi-armed quantum bandit problem $(\mathcal{A}, \Gamma = \mathcal{S}_d)$ for some $d \geq 2$ where the action set $\mathcal{A}$ is comprised of traceless observables. Then, the minimax regret satisfies
\begin{align}
	 R_n( \mathcal{A},  \mathcal{S}_d) = \Omega \left( \sqrt{n} \right)  
\end{align}
under some mild regularity conditions.
\end{result}
This is shown as Theorem~\ref{generalower} in Section~\ref{generalcase}. This result shows the dependence of the lower bound on the number of actions for a large class of action sets with mixed-state environments. The regularity condition serves to exclude sets where there is a dominant operator that performs optimally for all states, thus giving a trivial bound on the minimax regret. We also demand the observables to be traceless to aid the analysis.

For the special case of environments of one qubit quantum states and set of actions consisting of rank-1 projectors we give an exact form of the minimax regret lower bound.
In this case, the minimax regret can be bounded as
\begin{align}
	R_n (\mathcal{A} , \mathcal{S}_2) \geq\frac{ \sqrt{1-c }- (1-c)}{30} \sqrt{n}, 
\end{align}
where $c = \Tr ( \Pi_a \Pi_b) $ for the two observables that minimize the overlap $\max \{ \Tr ( \Pi_i\Pi_j ), \newline \Tr \big( \Pi_i(I-\Pi_j ) \big) \}$.
This is proved as Theorem \ref{teo:discrete1qubit} in Section \ref{sec:discrete1qubit}.

Next we want to study the dependence of the lower bound on the number of available actions and the dimension for suitable action sets.
For this purpose we consider action sets comprised of strings of Pauli observables (strings of single-qubit Pauli operators, excluding the overall identity).

\begin{result}\label{res:pauliobservables}
Consider a discrete multi-armed quantum bandit problem $(\mathcal{A}, \Gamma = \mathcal{S}_{d})$ for $d = 2^{\ell}$ and $\ell \in \mathbb{N}$, where the action set $\mathcal{A}$ is comprised of $k$ distinct length-$l$ strings of single-qubit Pauli operators. Then the minimax regret can be bounded as
\begin{align}
	 R_n (\mathcal{A}, \mathcal{S}_{d}) = \Omega \left(\sqrt{(k-1)n} \right) \,.
\end{align}
\end{result}
Note in particular that we can have $d^2-1$ such strings; thus, we get a lower bound of $R_n (\mathcal{A}, \mathcal{S}_{d}) = \Omega \big(d\sqrt{n}\big)$ for the most difficult of this type of quantum bandit problems. This result is stated as Theorem \ref{th:pauliobservables} in Section \ref{sec:Pauliobservables}. 
As we will see when we discuss the generic algorithm, this dependence is optimal.

We also study the specific case of one qubit rank-1 environments and show that even under these restricted environments the scaling of the regret with the number of rounds is the same.

\begin{result}\label{res:purepauli}
Consider a discrete multi-armed quantum bandit problem $(\mathcal{A}, \Gamma = \mathcal{S}^*_{2})$, where the action set $\mathcal{A}$ is comprised of the one qubit Pauli observables. Then the minimax regret can be bounded as
\begin{align}
	R_n (\mathcal{A}, \mathcal{S}^*_2 ) = \Omega \left( \sqrt{n} \right) \,.
\end{align}
\end{result}

This result is stated as Theorem \ref{th:pauliobservablesdiscrete} in Section \ref{sec:Pauliobservablesdiscrete}. 

\medskip

Now we move to general bandits with continuous action sets. In this case it is actually much more difficult to show lower bounds. One reason is that the sub-optimality gap is now not bounded away from zero when we choose a sub-optimal observable. We show that the scaling in the number of rounds is still the same also in this case.

\begin{result}\label{res:generalcontinous}
Consider a general multi-armed quantum bandit problem $(\mathcal{A}, \Gamma = \mathcal{S}_{d})$ where the action set $\mathcal{A}$ consists of all rank-1 projectors in $d$ dimensions.
Then, the minimax regret can be bounded as,
\begin{align} 
	R_n ( \mathcal{A}, \mathcal{S}_d) = \Omega( \sqrt{n}).
\end{align}
\end{result}
This result is shown in Theorem \ref{theor:AllPureStates} in Section \ref{sec:AllPureStates}. We would expect that the minimax regret necessarily also scales with the dimension, but we currently cannot show this. (See the discussion in Sections~\ref{sec:AllPureStates} and~\ref{sec:conclusion}.)

\subsection{Upper bounds}

In order to provide upper bounds for the minimax regret we explore algorithms and analyse the regret. For the general case we adapt one of the main algorithms for linear stochastic bandits called \textsf{LinUCB} to our quantum setting, which leads to the following result.

\begin{result}\label{generalalgorithm}
For any general multi-armed quantum bandit problem $(\mathcal{A}, \Gamma)$ with $\Gamma = \mathcal{S}_d$ or $\Gamma = \mathcal{S}^*_d$ , there exists a policy $\pi$ such that the worst-case regret can be bounded as
\begin{align}
	R_n ( \mathcal{A}, \Gamma, \pi) = O \left( d^2 \sqrt{n} \log n \right) \,. 
\end{align}
\end{result}
We give all the details in Section \ref{linucb}. Note that the scaling in terms of the numbers of rounds $n$ is quasi-optimal since it matches the previous lower bound up to logarithmic terms. However there is a gap in the dimensional dependence in comparison to the lower bound $\Omega (\sqrt{n})$ of Results \ref{res:generaldiscrete} and \ref{res:generalcontinous} or the lower bound $\Omega (d\sqrt{n})$ that we established for Pauli observables. The gap for Pauli observables can be closed if we consider a more generic algorithm called \textsf{UCB} that can be applied to any multi-armed stochastic bandit with a discrete number of arms, and thus can also treat linear bandits.

\begin{result}\label{res:UCB}
For any discrete multi-armed quantum bandit problem $(\mathcal{A}, \Gamma)$ with $\Gamma = \mathcal{S}_d$ or $\Gamma = \mathcal{S}^*_d$ such that $k = | \mathcal{A} |$, there exist policies $\pi_1,\pi_2$ such that the worst-case regret can be bounded as
\begin{align}
	R_n (\mathcal{A}, \mathcal{S}_d , \pi_1 ) = O \left( \sqrt{kn} \right), \quad R_n (\mathcal{A}, \mathcal{S}_d , \pi_2 ) = O \left( d\sqrt{n\log(nk)} \right).
\end{align}
\end{result}

The details for the above results can be found in Section \ref{sec:UCB}.

\medskip

Finally, we consider potential environments that contain all pure states in $\mathcal{S}_d^*$, again with an action set $\mathcal{A}$ containing all rank-1 projectors. For this case, the cumulative regret can be expressed as, $R_n(\mathcal{A}, \rho, \pi) = \sum_ {t=1}^n \EX_{\rho , \pi} (\frac12 \| \rho - \Pi_{A_t} \|_1 )^2$ where $\rho = \ket{\psi}\!\bra{\psi}$ is the unknown pure quantum state and $\Pi_{A_t}\in\mathcal{A}$ is a rank-1 projector selected at time step $t$. This regret now has a resemblance with the regret in online learning of quantum states since it is a natural loss function for our estimate $\Pi_{A_t}$ of the state $\rho$. However, the important difference is that in online learning the state estimate is provided separately from the choice of the next measurement, which is often chosen at random.

We are currently unable to provide a non-trivial lower bound on $R_n (\mathcal{A}, \mathcal{S}_d^*)$, and pose this as an open question. To get an upper bound we propose a policy that is based on a explore-then-commit strategy (see \cite{banditalgorithm}[Chapter 6)] ) where the exploration is done using a quantum state tomography algorithm. Using this policy we can show the following result.

\begin{result}\label{res:purealgorithm}
For any general multi-armed quantum bandit problem $(\mathcal{A}, \Gamma = \mathcal{S}_{2}^*)$ there exists a policy $\pi$ such that the  worst-case regret can be bounded as
\begin{align} 
	 R_n (\mathcal{A}, \Gamma,	\pi) = O\left( \sqrt{n}\log n  \right) .
\end{align}

\end{result}
 
We will see in Section \ref{purestatealgorithm} that this policy is ``simple'' in the sense that it is not very adaptive, suggesting room for improvements. Nevertheless it achieves the same type of scaling in the number of rounds that Result \ref{generalalgorithm}.

\section{Bretagnolle-Huber inequality and divergence decomposition lemma}

We will need to introduce some additional notation. For two probability measures $P, Q$ defined on the same probability space $(\Omega, \Sigma)$ we may introduce a (probability) measure $\mu$ that is dominating $P$ and $Q$ and define the Radon--Nikodym derivatives $p = \frac{dP}{d\mu}$ and $q=\frac{dQ}{d\mu}$. The choice of $\mu$ is arbitrary for the following definitions.
We define the Kullback--Leibler divergence as
\begin{align}
	D( P \| Q ) := \int p(\log p - \log q)\,d\mu
\end{align} 
where we use the convention that $0\log 0 = 0$ and $D(P\|Q) = \infty$ whenever $Q$ does not dominate $P$. We also define the (squared) Bhattacharyya coefficient as 
\begin{align}
	F(P, Q) = \left( \int \sqrt{p q} \, d\mu \right)^2 \,.
\end{align}
For two quantum states $\rho, \sigma \in \mathcal{S}_d$ we denote their quantum relative entropy as $D(\rho \| \sigma ) = \Tr(\rho \log \rho ) - \Tr(\rho \log \sigma)$ if $\text{supp}(\rho) \subseteq \text{supp}(\sigma)$ and $D(\rho \| \sigma ) = \infty$ otherwise. Their fidelity is defined as $F(\rho,\sigma) = ( \Tr |\sqrt{\rho}\sqrt{\sigma}|)^2$.

In this section we present some important technical results which we need in the subsequent proofs in order to prove the lower bounds summed up in the preceding section.
One of the technical tools that we will use in order to bound the regret is the following lemma due to Bertagnolle and Huber \cite{divergenceineq}.
\begin{lemma}[Bretagnolle--Huber inequality]\label{pinsker}
Let $P$ and $Q$ be probability measures on the same measurable space $(\Omega , \Sigma) $, and let $A\in \Sigma$ be an arbitrary event. Then,
\begin{align}
 P(A) + Q(A^c)  \geq \frac{1}{2} \exp ( -D(P \| Q) ) ,   
\end{align} 
where $A^c = \Omega \backslash A $ is the complement of $A$ and $D(P\|Q)$ is the Kullback–Leibler divergence.
\end{lemma}

We will use also a stronger bound that replaces $D(P \|Q)$ with the R\'enyi divergence of order $\alpha = \frac12$, defined as
\begin{align}
D_{\frac{1}{2}}(P \| Q) = - \log F(P, Q). 
\end{align}
We use the proof given in~\cite{banditalgorithm} for the Bretagnolle-Huber inequality in order to prove the alternative version.
\begin{lemma}\label{pinsker2}
Let $P$ and $Q$ be probability measures on the same measurable space $(\Omega , \Sigma) $, and let $A\in \Sigma$ be an arbitrary event with $A^c$ its complement. Then,
\begin{align}
 P(A) + Q(A^c)  \geq \frac{1}{2} \exp ( -D_{\frac{1}{2}}(P\|Q) ) = \frac12 F(P, Q).
\end{align} 
\end{lemma}
\begin{proof}
Let us define the probability measure $\mu:=\frac{1}{2}(P+Q)$ which dominates both $P$ and $Q$ and denote $dP/d\mu=:p$ and $dQ/d\mu=:q$. Note that,
\begin{align}
P(A) + Q(A^c) =&\int_A p\,\dd\mu + \int_{A^C} q\,d\mu \geq \int_A \min \lbrace p , q \rbrace\,\dd\mu + \int_{A^C} \min \lbrace p , q \rbrace\,\dd\mu\\
=& \int \min \lbrace p , q \rbrace\,\dd\mu.
\end{align}
In order to find a lower bound on this expression we use the Cauchy–Schwarz inequality to show that
\begin{align}
\frac{1}{2} \left( \int \sqrt{pq}\,d\mu \right)^2 
= \frac{1}{2} \left( \int \sqrt{\max \lbrace p , q \rbrace \min \lbrace p , q \rbrace}\,\dd\mu \right)^2 
& \\ \leq \frac{1}{2} \left( \int \max \lbrace p , q \rbrace\,\dd\mu  \right)\left( \int \min\lbrace p , q \rbrace\,\dd\mu \right) 
&\leq \int \min\lbrace p , q \rbrace\,\dd\mu ,
\end{align}
where in the final step we used that $ \int \max \lbrace p , q \rbrace \dd \mu \leq  \int (p + q) \dd \mu = 2$.
\end{proof}

The other main result that we will need allows us to decompose the divergence computed for two joint distribution that result from the same policy applied to two different quantum states.

\begin{lemma}[Divergence decomposition lemma]\label{divergence}
Let $\mathcal{A}=\left\lbrace O_1,...,O_k \right\rbrace$ be a set of actions and $\rho$ and $\rho'$ two quantum states defining two multi-armed quantum bandits with action set $\mathcal{A}$. Fix some policy $\pi$
and let $P_{\rho, \pi }$ and $ P_{\rho' \pi }$ be the probability distributions induced by the $n$-round interconnection of $\pi $ and $\rho $ described in~\eqref{probdens}. Then,
\begin{align}
	D(P_{\rho, \pi}  \| P_{\rho' , \pi}) = \sum_{a=1}^k \EX_{\rho,\pi} [T_a (n ) ] D \big( P_{\rho}(\cdot | a) \big\| P_{\rho'}(\cdot | a) \big) .
\end{align}
\end{lemma}

The above lemma and the proof can be found in \cite[Chapter 15]{banditalgorithm} for the classical model of multi-armed stochastic bandits. We have restated the lemma for our quantum case but this statement and the proof follows trivially from the original one. The proof is a consequence of the chain rule for the KL divergence. Unfortunately, there is no such decomposition for R\'enyi divergences but we will see in Section \ref{lower} that applying the data processing inequality and bounding it with the sandwiched quantum R\'enyi divergences (see, e.g.,~\cite{lennert13,wilde13}) between $\rho$ and $\rho'$ is sufficient to bound the regret.

For this purpose, we present result on the divergences between the probability distributions $P_{\rho,\pi}$. Note that this result also holds in the general bandit case where the set of arms can be continuous. What we mean below by `quantum extension' of a classical relative entropy $D$ is that, if quantum states $\rho$ and $\sigma$ commute, then $\tilde{D}(\rho\|\sigma)=D(p\|q)$ where $p$ and, respectively, $q$ are the vectors of eigenvalues of $\rho$ and, respectively, $\sigma$, in a common eigenbasis.

\begin{lemma}\label{lemma:D1/2ineq}
Let $\alpha\in\R$ be such that the classical R\'{e}nyi relative entropy $D_\alpha$ can be given a quantum extension (which we denote with the same symbol) which is additive and satisfies the data processing inequality. For any policy $\pi$ and environment states $\rho,\,\rho'$, we have
\begin{align}\label{eq:D1/2ineq}
D_\alpha(P_{\rho,\pi}\|  P_{\rho',\pi})\leq nD_\alpha(\rho\|\rho').
\end{align}
\end{lemma}
In particular, for $\alpha=1/2$ we can choose $D_{\frac12}(\rho\|\rho') = - \log{F(\rho,\rho')}$ and, for $\alpha=1$, we can let the quantum extension $D \equiv D_1$ of the Kullback-Leibler relative entropy be the quantum relative entropy.

\begin{proof}
Let us fix the policy $\pi$ and use the notations and definitions of Section \ref{subsec:generalbandits}. We prove the claim by constructing a positive-operator-valued measure (POVM) $\ms E$ over the value space $\mc C:=(\Sigma\times\mc X)^n$ and operating in $\hil^{\otimes n}$ such that
\begin{align}\label{eq:policyPOVM}
P_{\sigma,\pi}(A_1,x_1,\ldots,A_n,x_n)=\tr{\sigma^{\otimes n}\ms E(A_1,x_1,\ldots,A_n,x_n)}
\end{align}
for all states $\sigma$, and $A_t\in\Sigma$ and $x|t\in\mc X$ for $t=1,\ldots,\,n$. Note that, as $\ms E$ is not a set function, it is, strictly speaking, not a POVM in the same sense as $P_{\sigma,\pi}$ is not a measure. However, this should cause no confusion. Using the data-processing inequality and the additivity of $D_\alpha$, we have, for all states $\rho$ and $\rho'$,
\begin{align}
D_\alpha(P_{\rho,\pi}\|P_{\rho',\pi})\leq D_\alpha(\rho^{\otimes n}\|\rho'^{\otimes n})=nD_\alpha(\rho\|\rho'),
\end{align}
implying Inequality \eqref{eq:D1/2ineq}. (Note that the POVM $\ms E$ corresponds to a quantum-to-classical channel which maps $\sigma^{\otimes n}$ into $ P_{\sigma,\pi}$). Thus all that remains is to write down $\ms E$.

Recall the reward distributions $P_\sigma(x|a)$ in state $\sigma$ conditioned by the arm $a\in\mc A$. Through linear extension, we may define these distributions $P_\tau(x|a)$ for any (trace-class) operators $\tau$; these are, in general, complex distributions. We may define the POVM $\ms E$ through
\begin{align}
&\tr{(\tau_1\otimes\cdots\otimes\tau_n)\ms E(A_1,x_1,\ldots,A_n,x_n)} \nonumber \\
=&\int_{A_1}\cdots\int_{A_n}\prod_{t=1}^nP_{\tau_t}(x_t|a_t)\pi_t(da_t|a_1,x_1,\ldots,a_{t-1},x_{t-1})\label{eq:EPOVM}
\end{align}
for all (trace-class) operators $\tau_t$ and $A_t\in\mc B$ and $x_t\in\mc X$, $t=1,\ldots,n$. As $\tau\mapsto P_\tau(x|a)$ is a linear functional and $x\mapsto P_\sigma(x|a)$ is a probability distribution whenever $\sigma$ is a state, it easily follows that there are POVMs $\ms P_a$ on $\mc X$ operating in $\hil$ such that $P_\tau(x|a)=\tr{\tau\ms P_a(x)}$. Equivalently (although slightly less formally) we may now write $\ms E$ in the differential form
\begin{align}
\ms E(da_1,x_1,\ldots,da_n,x_n)=\bigotimes_{t=1}^n\pi_t(da_t|a_1,x_1,\ldots,a_{t-1},x_{t-1})\ms P_{a_t}(x_t).
\end{align}
By substituting $\tau_t=\sigma$ for all $t=1,\ldots,n$ in Equation \eqref{eq:EPOVM}, we immediately see that $P_{\sigma,\pi}(A_1,x_1\ldots,A_n,x_n)=\tr{\sigma^{\otimes n}\ms E(A_1,x_1\ldots,A_n,x_n)}$ for all $A_t\in\Sigma$ and $x_t\in\mc X$, $t=1,\ldots,n$, and, thus, Equation \eqref{eq:policyPOVM} holds.
\end{proof}

\section{Regret lower bounds for discrete bandits}\label{lower}

In this section we will focus on proving mini-max regret lower bounds for different sets of environments using information-theoretic tools. 

\subsection{General case}\label{generalcase}

The first case that we study are quantum multi-armed bandits where the environment can be any state in a Hilbert space of dimension $d$, and the action set any set of observables. We will give a lower bound for the minimax regret, and in order to do it we construct a quantum state such that for every policy it achieves a non-trivial lower bound. Note that if we do not impose any condition on the action set the lower bound will vanish. The reason is that if there is an operator in the action set that dominates over the others, the policy that always chooses this operator at each round achieves 0 regret. More specifically, suppose that $\mathcal{A} = \left\lbrace O_1 , O_2 \right\rbrace$ with $O_1\geq O_2$. Then we know that independently of the environment $\rho$, $\Tr(\rho O_1 ) \geq \Tr(\rho O_2 )$ and the policy that always chooses $O_1$ will always pick the optimal action.

So, we will impose a condition on the action set that ensures that there is no such dominant action/operator. The condition is the following:
in the action set $\mathcal{A}$ there exist at least two operators $O_a,O_b\in \mathcal{A}$ with maximal eigenvectors $\ket{\psi_A},\ket{\psi_B}$ such that for any $i\neq a$ and $j\neq b$ 
\begin{align}\label{condition}
\bra{\psi_A } O_a \ket{\psi_A} > \bra{\psi_A } O_i \ket{\psi_A}\quad \text{and}\quad
\bra{\psi_B } O_b \ket{\psi_B} > \bra{\psi_B } O_j \ket{\psi_B}.
\end{align}

\begin{theorem}\label{generalower}
Let $n\in\mathbb{N}$. For any policy $\pi$ and action set of traceless observables $\mathcal{A}$ that obeys condition \eqref{condition} there exists an environment $\rho\in\mathcal{S}_d $ such that,
\begin{align}
 R_n(\mathcal{A},\rho,\pi) \geq C_{\mathcal{A}}\sqrt{n} ,
 \end{align}
where $C_{\mathcal{A}}>0$ is a constant that depends on the action set.
\end{theorem}

\begin{proof}
Choose two operators $O_a,O_b\in\mathcal{A}$ that obey the condition in~\eqref{condition} with maximal eigenvectors $\ket{\psi_A}$ and $\ket{\psi_B}$, respectively. Define the following environments:
\begin{align}
\rho := \frac{1-\Delta}{d}I + \Delta \ket{\psi_A}\bra{\psi_A},\quad \rho' := \frac{1-\Delta}{d}I + \Delta \ket{\psi_B}\bra{\psi_B},
\end{align}
for some constant $0\leq \Delta \leq \frac{1}{2}$ to be defined later. Note that $\rho$ and $\rho'\geq 0$. Define

\begin{align}
c = \min \left\lbrace \bra{\psi_A } O_a \ket{\psi_A} - \bra{\psi_A } O_i \ket{\psi_A},
\bra{\psi_B } O_b \ket{\psi_B} - \bra{\psi_B } O_j \ket{\psi_B} \right\rbrace,
\end{align}
for $i\neq a$ and $j\neq b$. Using the expression for the regret in Equation~\eqref{regretarms}, we can compute the regret for $\rho$ as,
\begin{align}
R_n(\mathcal{A}, \rho , \pi) = \sum_{i=1}^k \EX_{\rho,\pi}(T_i(n))\Delta \left( \bra{\psi_A } O_a \ket{\psi_A} - \bra{\psi_A } O_i \ket{\psi_A} \right)
\end{align}
where we have used that the observables are traceless and 
\begin{align}
\max_{O_i\in \mathcal{A}} \Tr (\rho O_i ) = \Tr(\rho O_a ) = \Delta \bra{\psi_A } O_a \ket{\psi_A}.
\end{align}
Note that for $i=a$ the sub-optimality gaps Equation~\eqref{suboptimality} are zero. Thus, using Condition \eqref{condition} we can bound the regret as,
\begin{align}
R_n(\mathcal{A},\rho , \pi ) = \sum_{i\neq a} \EX_{\rho,\pi}[T_i(n)]\Delta \left( \bra{\psi_A } O_a \ket{\psi_A} - \bra{\psi_A } O_i \ket{\psi_A} \right)\geq \Delta c \sum_{i\neq a} \EX_{\rho,\pi} [T_i(n)].
\end{align}
Using that $n=\sum_{i=1}^k \EX_{\rho,\pi}[T_i(n)]$ and Markov inequality we have,
\begin{align}\label{rega}
R_n(\mathcal{A},\rho , \pi) \geq \Delta c \EX_{\rho,\pi}  \left[ n - T_a(n) \right] \geq \frac{cn\Delta}{2}P_{\rho,\pi}\left( T_a(n)\leq \frac{n}{2} \right).
\end{align}
Similarly, the regret for $\rho'$ can be bounded as,
\begin{align}
R_n(\mathcal{A},\rho',\pi) \geq \Delta \EX_{\rho',\pi}[T_a(n)]\left(  \bra{\psi_B } O_b \ket{\psi_B} - \bra{\psi_B } O_i \ket{\psi_B} \right),
\end{align}
where we have taken into account just the term with $i=a$. Using Condition \eqref{condition} and Markov inequality we have,
\begin{align}\label{regb}
R_n(\mathcal{A},\rho',\pi ) \geq \frac{cn\Delta}{2}P_{\rho' , \pi} \left( T_a(n) > \frac{n}{2} \right).
\end{align}
Thus, combining Equations~\eqref{rega} and \eqref{regb},
\begin{align}
R_n(\mathcal{A},\rho ,\pi ) + R_n(\mathcal{A},\rho' , \pi ) \geq \frac{cn\Delta}{2}\left( P_{\rho,\pi}\left( T_a(n)\leq \frac{n}{2} \right) + P_{\rho' , \pi} \left( T_a(n) > \frac{n}{2} \right) \right).
\end{align}
Using Lemma \ref{pinsker} we can bound the above expression as,
\begin{align}\label{regb2}
R_n(\mathcal{A},\rho ,\pi ) + R_n(\mathcal{A},\rho' , \pi ) \geq \frac{cn\Delta}{4}\exp\left( -D(P_{\rho , \pi}  \| P_{\rho' , \pi}  ) \right).
\end{align}
Using Lemma \ref{divergence} combined with data-processing inequality we can bound the Kullback-divergence as,
\begin{align}
D(P_{\rho , \pi}  \| P_{\rho' , \pi}  ) =& \sum_{a=1}^k \EX_{\rho , \pi } (T_a(n) )  D \big( P_{\rho}(\cdot | a) \big\| P_{\rho'}(\cdot | a) \big) )\leq D(\rho \| \rho' )\sum_{a=1}^k \EX_{\rho' , \pi } [T_a(n) ]\nonumber\\
=& nD(\rho \| \rho' ).\label{brelent}
\end{align}
where $D(\rho \| \rho' )= \Tr\rho\log\rho-\Tr\rho\log\rho'$ is the relative entropy between $\rho$ and $\rho'$.
Now define $f(\Delta) = D(\rho \| \rho' )$ for the corresponding expressions of $\rho$ and $\rho'$ and note that $f(0) = 0$ since for $\Delta=0$, $\rho = \rho'$ and $f'(0)=0$ using the convexity of the quantum relative entropy. Thus, using Taylor's theorem we can express $f(\Delta )$ as,
\begin{align}
f(\Delta ) = \frac{\Delta^2}{2}f''(\chi ),
\end{align}
for some $\chi \in [0,\Delta ] $. Define, 
\begin{align}
c_f = \max f''(\chi ) \quad \text{for}\quad \chi\in [0,\frac{1}{2} ].
\end{align}
Note that $f''(\chi)$ is well defined since the states $\rho,\rho'$ have full support and then $f(\chi )$ is a smooth function. Thus using that $f''(\chi)\leq c_f$ for $\Delta\in [0,\frac{1}{2} ]$ we can plug the above expression into Equation~\eqref{brelent} and using Equation~\eqref{regb2} we have,
\begin{align}
R_n(\mathcal{A},\rho ,\pi ) + R_n(\mathcal{A},\rho' , \pi ) \geq \frac{cn\Delta}{4}\exp\left( -\frac{n\Delta^2}{2} c_f\right).
\end{align}
Finally, if we choose $\Delta = \frac{1}{2\sqrt{n}}$, 
\begin{align}
R_n(\mathcal{A},\rho ,\pi ) + R_n(\mathcal{A},\rho' , \pi ) \geq \frac{c}{8}\exp\left( -\frac{1}{8} c_f \right)\sqrt{n},
\end{align}
and the result follows using $2\max \left\lbrace R_n(\mathcal{A},\rho ,\pi ) , R_n(\mathcal{A},\rho' , \pi) \right\rbrace \geq R_n(\mathcal{A},\rho ,\pi ) + R_n(\mathcal{A},\rho' , \pi )$.
\end{proof}

\subsection{Special case: one qubit and rank-1 projectors}\label{sec:discrete1qubit}
Now we study the specific case where the environment belongs to the set of one-qubit quantum states and the set of actions is a set of rank-1 projectors, i.e $\mathcal{A}=\left\lbrace \Pi_1,...,\Pi_k \right\rbrace $ where $\Pi_i$ are rank-1 matrices such that $\Pi_i^2 = \Pi_i$. For this case we are able to give the exact constant of Theorem \ref{generalower}. 

Given a quantum state $\rho$ we will use the Bloch state representation,

\begin{align}\label{qbitdensity}
\rho = \frac{I}{2} + \frac{1}{2} \bold{r}\cdot \sigma,
\end{align}
where $\bold{r} \in \mathbb{R}^3$ is the Bloch vector and $\sigma = \left( \sigma_x , \sigma_y , \sigma_z \right)$ is the Pauli vector. Recall that for mixed states $\| \bold{r} \|_2 < 1$ and for pure states $\| \bold{r} \|_2 =1$. The projectors $\Pi_i$ of the action set can be written as \eqref{qbitdensity} for some unit vector $\bold{r}_i\in \mathbb{R}^3$. So, the action set is characterized by the set of measurement directions $\mathcal{A}_r = \lbrace \bold{r}_1,...,\bold{r}_k \rbrace$ and the mean of each action is,
\begin{align}\label{expectation}
\Tr \left( \rho \Pi_t \right) = \frac{1}{2} + \frac{1}{2} \bold{r}\cdot \bold{r}_t.
\end{align}

In order to prove the main result we will need the computation of the relative entropy between two quantum states given by the following Lemma.
\begin{lemma}\label{relativeentropy}
Let $\rho = \frac{I}{2} + \frac{\Delta}{ 2} \sigma_x$ and $\rho '= \frac{I}{2} + \frac{\Delta}{ 2} \sigma_z$ be two one-qubit density matrices. Then, their quantum relative entropy can be computed as,
\begin{align}
	D(\rho \| \rho' ) = \frac{\Delta}{2} \log \left( \frac{1+\Delta}{1-\Delta} \right).
\end{align}
\end{lemma}

\begin{proof}
Let $\Pi^+_i , \Pi^-_i$ be the projectors for the $i=x,z$ Pauli matrix into the subspaces of eigenvalue $+1$ and $-1$ respectively. Then we can express the density matrices as,
\begin{align}\label{roproj}
\rho = \frac{1+\Delta}{2}\Pi^+_x + \frac{1-\Delta}{2}\Pi^-_x, \quad \rho'  = \frac{1+\Delta}{2}\Pi^+_z + \frac{1-\Delta}{2}\Pi^-_z,
\end{align}
where we have used that $ \sigma_i = \Pi^+_i - \Pi^-_i   $. Then using that $ I =  \Pi^+_i + \Pi^-_i$, $\Tr(\sigma_i)=0$ and $\Tr \sigma_x \sigma_z = 0$ we can compute the following identity
\begin{align}\label{trproj}
\Tr (\Pi^+_i \Pi^-_j) = \left\lbrace\begin{array}{c} \frac{1}{2} \quad \text{if } i\neq j \\ 0 \quad \text{if } i=  j. \end{array}\right.
\end{align}
Then, using Equations~\eqref{roproj} and \eqref{trproj},
\begin{align}
\Tr \rho \log \rho = \frac{1+\Delta}{2}\log\left( \frac{1+\Delta}{2} \right) + \frac{1-\Delta}{2}\log\left( \frac{1-\Delta}{2} \right),
\end{align}
\begin{align}
\Tr \rho \log \rho' = \frac{1}{2}\log\left( \frac{1+\Delta}{2} \right) + \frac{1}{2}\log\left( \frac{1-\Delta}{2} \right).
\end{align}
The result follows from the definition of the relative entropy and rearranging the above terms.

\end{proof}

\begin{theorem}\label{teo:discrete1qubit}
Let $n\in\mathbb{N}$. For any policy $\pi $ and finite action set of observables $\mathcal{A}$ of rank-1 projectors, there exists a 1-qubit environment $\rho\in\mathcal{S}_2 $ such that,
\begin{align}
	R_n (\mathcal{A},\rho,\pi ) \geq\frac{ \sqrt{1-c} - (1-c) }{30} \sqrt{n}, 
\end{align}
where $c = \Tr ( \Pi_a \Pi_b) $ and $(\Pi_a,\Pi_b) = \argmin_{\Pi_i,\Pi_j\in\mathcal{A}} \max \big\{ \Tr ( \Pi_i\Pi_j ), \Tr \big( \Pi_i(I-\Pi_j ) \big) \big\} $.
\end{theorem}
Note that that $\sqrt{1-c} \geq 1-c$ since $c \in [0,1]$, and the constant is thus positive and only vanishes for $c \in \{0,1\}$. 

\begin{proof}
Let $\mathcal{A}_r=\left\lbrace \bold{r}_1,...,\bold{r}_k \right\rbrace$ be the set of vectors associated to the $k$ elements of $\mathcal{A}$.  Let $\Delta \in [0,\frac{1}{2}]$ be a constant to be chosen later. Pick $\bold{r}_a,\bold{r}_b\in \mathcal{A}_r$ such that they are the two directions with smallest inner product, that is
\begin{align}
( \bold{r}_a , \bold{r}_b ) &= \argmin_{\bold{r}_i,\bold{r}_j \in \mathcal{A}_r} \vert \bold{r}_i\cdot\bold{r}_j \vert 
= \argmin_{\Pi_i,\Pi_j\in\mathcal{A}} \left| \Tr ( \Pi_i\Pi_j ) - \Tr \big( \Pi_i(I-\Pi_j ) \big) \right| \\
&= \argmin_{\Pi_i,\Pi_j\in\mathcal{A}} \max \big\{ \Tr ( \Pi_i\Pi_j ) - \Tr \big( \Pi_i(I-\Pi_j ) \big) , \Tr \big( \Pi_i(I-\Pi_j ) \big) - \Tr ( \Pi_i\Pi_j ) \big\} \\
&= \argmin_{\Pi_i,\Pi_j\in\mathcal{A}} \max \big\{ \Tr ( \Pi_i\Pi_j ), \Tr \big( \Pi_i(I-\Pi_j ) \big) \big\} \,,
\end{align}
where in the last step we used that $ \Tr ( \Pi_i\Pi_j ) + \Tr \big( \Pi_i(I-\Pi_j ) \big) = 1$.
Since $\mathcal{A}$ contains at least two independent directions, $\bold{r}_a, \bold{r}_b$ determine a plane. Now rotate $\bold{r}_a , \bold{r}_b$ symmetrically in this plane until we find two orthogonal directions $\bold{r}'_a$ and $\bold{r}'_b$ that obey the following conditions:

\begin{itemize}
\item $\vert  \bold{r}'_a \vert = \vert \bold{r}'_b \vert  = 1$ and $\bold{r}'_a \cdot \bold{r}'_b=0$ (unit orthogonal vectors).
\item $\bold{r}_a\cdot \bold{r}'_a = \bold{r}_b\cdot \bold{r}'_b$ (symmetrical rotation).
\item $\bold{r}_a \times \bold{r}_b = \bold{r}'_a \times \bold{r}'_b$ (they remain in the same plane).
\item $\max_{\bold{r}_i\in \mathcal{A}} \bold{r}'_a\cdot \bold{r}_i = \bold{r}'_a\cdot \bold{r}_a$, $\max_{\bold{r}_i\in \mathcal{A}} \bold{r}'_b\cdot \bold{r}_i = \bold{r}'_b\cdot \bold{r}_b$ (closest directions that obey the above conditions).

\end{itemize}

\begin{figure}
\centering
\begin{overpic}[percent,width=0.4\textwidth]{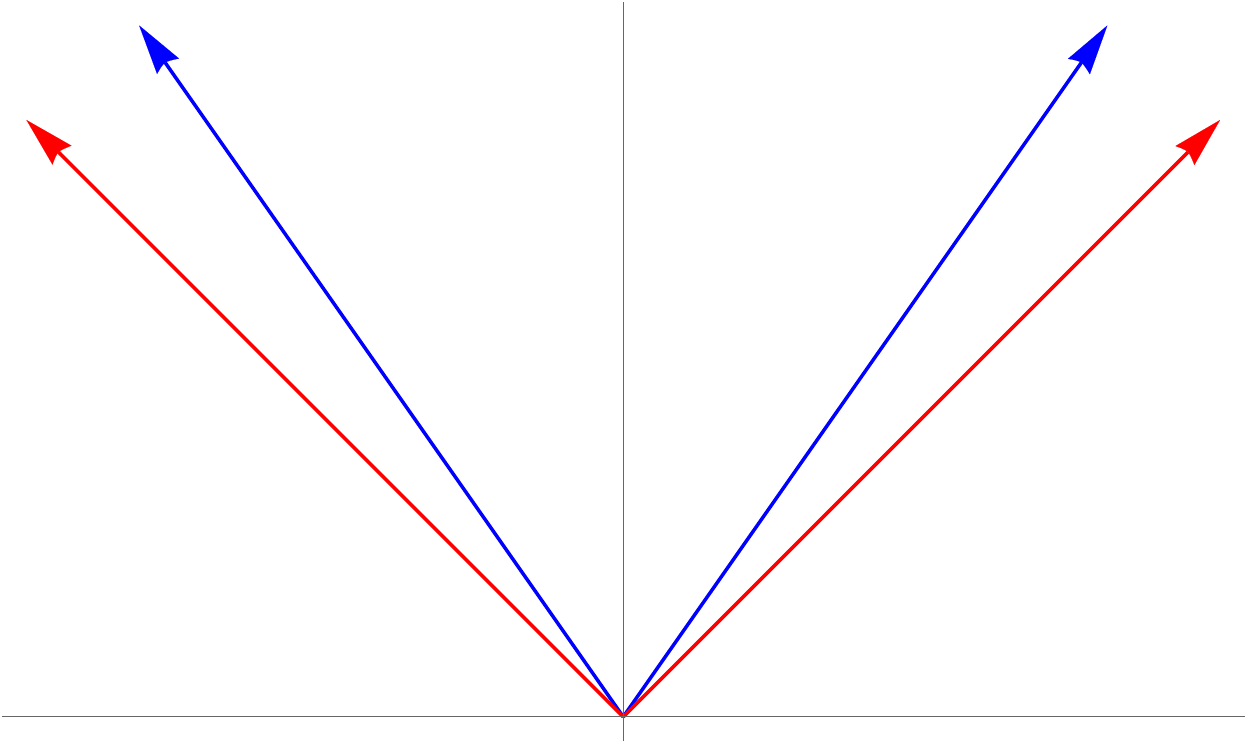}
\put(7,58){\color{blue} $\mathbf{r}_a$}
\put(89,58){\color{blue} $\mathbf{r}_b$}
\put(1,51){\color{red} $\mathbf{r}'_a$}
\put(96,51){\color{red} $\mathbf{r}'_b$}
\end{overpic}
\caption{Scheme for the choice of the vectors $\mathbf{r}'_a$ and $\mathbf{r}'_b$.}
\label{fig:scheme2}
\end{figure}
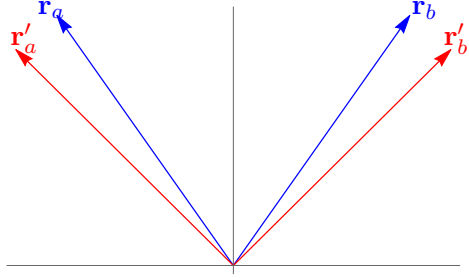

Using the directions $\bold{r}'_a$ and $\bold{r}'_b$ define the following environments $\rho_a , \rho_b$
\begin{align} 
\rho_a := \frac{I}{2} + \frac{\Delta}{2}\bold{r}'_a\cdot \sigma \quad \rho_b := \frac{I}{2} + \frac{\Delta}{2}\bold{r}'_b \cdot \sigma .
\end{align}

Note that $\rho\geq 0$ and $\rho_b\geq 0$. In order to calculate the sub-optimality gaps~\eqref{suboptimality} for $\rho$ and $\rho_b$ we will use Equation~\eqref{expectation} and an orthonormal system of coordinates defined by the directions $\mathcal{C}_{a,b}= \left\lbrace \bold{r}'_a , \bold{r}'_b ,\bold{r}'_a \times \bold{r}'_b \right\rbrace$. Using the above conditions for $\bold{r}'_a $ and $\bold{r}'_b$ we have
\begin{align}\label{suboptimalityab}
\Delta^a_i = \Tr (\rho\Pi_i) = \Delta ( p - r_{ia} ), \quad \Delta^b_i= \Tr(\rho_b\Pi_b ) = \Delta ( p - r_{ib} ) \quad \text{for } i = 1,...,k,
\end{align}
where $p = \bold{r}_a\cdot \bold{r}'_a$ and $r_{ia} , r_{ib}$ are the first and second coordinates of $\bold{r}_i \in \mathcal{A}$ in the coordinate system  $\mathcal{C}_{a,b}$ respectively. This last fact is because we defined $\mathcal{C}_{a,b}$ to be a orthonormal set of coordinates. Note that by the construction of $\bold{r}'_a$ and $\bold{r}'_b$ we have
\begin{align}
p = \cos \left( \frac{\pi}{4} - \frac{\theta}{2}\right),
\end{align}
where $\theta$ is defined via $\cos\theta = \bold{r_a}\cdot \bold{r_b}$. Let $\Pi_a,\Pi_b$ be the projectors associated to $\bold{r}_a,\bold{r}_b$, then using the trigonometric identity for $\cos\left(\frac{a}{2}\right)$ we have $\Tr(\Pi_a\Pi_b ) = \cos^2\left(\frac{\theta}{2} \right)$. Thus, using the trigonometric identity for $\cos(a+b)$ we have
\begin{align}\label{constantp}
p = \frac{1}{\sqrt{2}}\left( \sqrt{\Tr(\Pi_a \Pi_b}) + \sqrt{1-\Tr(\Pi_a \Pi_b})\right).
\end{align}

Now define the following subsets of indices of $[k]$,
\begin{align}
\mathcal{N}_a = \left\lbrace i \colon r^2_{ia} \geq \frac{1}{2} \text{ for } \bold{r}_i \in \mathcal{A}_r\right\rbrace \quad \mathcal{N}_a^C = \left\lbrace i \colon r^2_{ia} < \frac{1}{2} \text{ for } \bold{r}_i \in \mathcal{A}_r \right\rbrace.
\end{align}

Note that this sets are complementary, $\mathcal{N}_a \cap \mathcal{N}_a^C =\emptyset$ and $\mathcal{N}_a \cup \mathcal{N}_a^C = [k]$.

Using the expression for the regret~\eqref{regretarms} and the sub-optimality gap~\eqref{suboptimalityab}, the regret for $\rho $ can be bounded as,

%

\begin{align} 
	R_n (\mathcal{A}, \rho_a , \pi) &= \Delta \sum_{i=1}^k ( p - r_{ia} ) \EX_{\rho_a, \pi} [T_i (n ) ] \\
	&\geq \Delta \sum_{i\in \mathcal{N}_a^C} ( p - r_{ia} ) \EX_{\rho_a , \pi} [T_i (n ) ]   \\
	&\geq \Delta  \left( p - \frac{1}{\sqrt{2}} \right) \sum_{i\in   \mathcal{N}_a^C} \EX_{\rho_a , \pi} [T_i (n ) ], 
\end{align}

where we have used that for $i\in  \mathcal{N}_a^C$, we have $r_{ia} \leq \frac{1}{\sqrt{2}}$. Using that $n = \sum_{i=1}^k  \EX_{\pi, \rho_a} (T_i (n ) ),$
\begin{align}
 R_n (\mathcal{A}, \rho_a , \pi) \geq \Delta  \left( p - \frac{1}{\sqrt{2}}\right) \EX_{\pi , \rho_a} \left[ n - \sum_{i\in   \mathcal{N}_a} T_i (n ) \right]. 
\end{align}
Applying Markov's inequality to the above expression we have,
\begin{align} 
R_n ( \mathcal{A},\rho_a , \pi ) \geq \frac{n \Delta}{2} \left( p - \frac{1}{\sqrt{2}} \right) P_{ \pi ,\rho_a  } \left(   n - \sum_{i\in   \mathcal{N}_a} T_i (n ) \geq \frac{n}{2} \right).   
\end{align}
Rearranging all the terms,
\begin{align}\label{regret1rho}
 R_n ( \mathcal{A},\rho_a , \pi )  \geq \frac{n \Delta}{2} \left( p - \frac{1}{\sqrt{2}} \right) P_{\pi , \rho_a } \left(  \sum_{i\in   \mathcal{N}_a} T_i (n ) \leq \frac{n}{2} \right).
\end{align}
In order to bound the regret for $\rho_b$ note that if $i \in \mathcal{N}_a$ then $r^2_{ib} \leq \frac{1}{2}$ since $\| \bold{r}_i \| = 1$. Using the same tricks as before,
\begin{align} 
	R_n (\mathcal{A}, \rho_b , \pi )  &\geq \Delta \sum_{i\in \mathcal{N}_a} ( p - r_{ib} ) \EX_{\pi , \rho_b } [T_i (n ) ]\\ 		&\geq \Delta  \left( p - \frac{1}{\sqrt{2}} \right) \sum_{i\in   \mathcal{N}_a} \EX_{\pi , \rho_b} [T_i (n )]. 
\end{align}
Using again Markov's inequality,
\begin{align}\label{regret2rho}
R_n (\mathcal{A}, \rho_b , \pi ) \geq \frac{n \Delta}{2} \left( p - \frac{1}{\sqrt{2}} \right) P_{\pi , \rho_b } \left(  \sum_{i\in   \mathcal{N}_a} T_i (n ) > \frac{n}{2} \right).
\end{align}
Thus, combining Equations~\eqref{regret1rho} and \eqref{regret2rho},
\begin{align}
R_n ( \mathcal{A},\rho_a , \pi )  + R_n ( \mathcal{A}, \rho_b , \pi ) \geq   \nonumber
\end{align}
\begin{align}
\frac{n \Delta}{2} \left( p - \frac{1}{\sqrt{2}} \right) \left(   P_{\pi , \rho_a  } \left(  \sum_{i\in   \mathcal{N}_a} T_i (n ) \leq \frac{n}{2} \right) +  P_{\pi , \rho_b } \left(  \sum_{i\in   \mathcal{N}_a} T_i (n ) > \frac{n}{2} \right) \right).
\end{align}
Using Lemma \ref{pinsker} we can bound the above expression as,
\begin{align}\label{sumregret}
R_n ( \rho_a , \pi,\mathcal{A} ) + R_n (\mathcal{A} , \rho_b , \pi) \geq  \frac{n \Delta}{4} \left( p - \frac{1}{\sqrt{2}}\right)\exp \left(- D(P_{\pi , \rho_a  } \| P_{\pi , \rho_b}) \right).
\end{align}
Using Lemma \ref{divergence} combined with data-processing inequality we can bound the Kullback–Leibler divergence as,
\begin{align}
 D(P_{\pi , \rho_a  } \| P_{\pi , \rho_b}) &= \sum_{i=1}^k \EX_{\pi , \rho_a  } [T_i (n ) ] D \big( P_{\rho_a}(\cdot | i) \big\| P_{\rho_b}(\cdot | i) \big) ) \\
 &\leq D ( \rho_a \| \rho_b ) \sum_{i=1}^k \EX_{\rho , \pi } [T_i (n ) ] = n  D ( \rho_a \| \rho_b )  .
\end{align}

Using that the relative entropy $D ( \rho_a  \| \rho_b )$ is unitarily invariant and $\bold{r}'_a,\bold{r}'_b$ are orthogonal we can use the computation of the relative entropy given by Lemma \ref{relativeentropy}. Thus we can lower bound Equation~\eqref{sumregret} as,
\begin{align}
R_n ( \mathcal{A},\rho_a , \pi )  + R_n (\mathcal{A} , \rho_b , \pi) \geq \frac{n \Delta}{4} \left( p - \frac{1}{\sqrt{2}}\right)\exp \left(  - \frac{n\Delta}{2} \log \left( \frac{1+\Delta}{1-\Delta} \right) \right).
\end{align}
Choosing $\Delta = \frac{1}{2\sqrt{n}}$ we have the following upper bound,
\begin{align}
 \log \left( \frac{1+\Delta}{1-\Delta} \right) \leq 2 \log (3 ) \Delta \quad \text{for }\Delta \in [0,1/2],
\end{align}
where we have used that $\log \left( \frac{1+x}{1-x} \right) $ is convex and finite for $x\in [ 0,\frac{1}{2} ]$. Thus, choosing $\Delta = \frac{1}{2\sqrt{n}}$ we have

\begin{align}
 R_n ( \mathcal{A},\rho_a , \pi )  + R_n (\mathcal{A} , \rho_b , \pi ) \geq  \frac{ p - \frac{1}{\sqrt{2}}}{8\cdot 3^{1/4}} \sqrt{n},
\end{align}
and the result follows using the expression for $p$ \eqref{constantp} and $15 > 8\cdot 3^{1/4}\sqrt{2}$. From Equation \eqref{constantp} is easy to check that $p\geq \frac{1}{\sqrt{2}}$, thus the above lower bound is positive.
\end{proof}

\subsection{Pauli observables}\label{sec:Pauliobservables}

The next set of actions that we study are Pauli observables. If we consider a $d=2^m-$dimensional ($m$ qubits) Hilbert space, there are $d^2$ different Pauli observables and they can be expressed as the $m$-fold tensor product of the $2\times 2$ Pauli matrices. Let $\sigma_1,...,\sigma_{d^2}$ denote all the possible Pauli observables. Each $\sigma_i$ can be expressed as $\sigma_i = \Pi^+_i - \Pi^-_i$, where $\Pi^+_i ,\Pi^-_i$ are projectors associated to the $+1$ and $-1$ subspaces and they describe the 2 possible outcomes when we perform measurements using Pauli observables. 
\begin{theorem}\label{th:pauliobservables}
Let $n\in\mathbb{N}$. For any policy $\pi$ and action set of observables $\mathcal{A}$ comprised of $k$ distinct length-$m$ strings of single-qubit Pauli observables for $d=2^m$ with $m\in\mathbb{N}$, there exists an environment $\rho \in \mathcal{S}_d$ such that
\begin{align}
	R_n(\mathcal{A},\rho , \pi ) \geq \frac{3}{100}	\sqrt{(k-1)n},
\end{align}
\end{theorem}
for $n\geq 2( k-1)$.

\begin{proof}
The case $k = 1$ is trivial since $R_n = 0$ always. Suppose $k>1$ and let $0\leq \Delta \leq \frac{1}{3}$ be a constant to be chosen later. Pick $\sigma_1 \in \mathcal{A}$ and define the following environment,
\begin{align}
\rho := \frac{I}{d} + \frac{\Delta}{d}\sigma_1.
\end{align}
Define
\begin{align}\label{lessplayed}
l := \argmin_{j>1} \EX_{\rho,\pi} \left[ T_j (n)  \right]
\end{align}
as the index for the least expected picked observable different from $\sigma_1$. Define a second environment as,
\begin{align}
\rho' := \frac{I}{d} + \frac{\Delta}{d}\sigma_1 + \frac{2\Delta}{d}\sigma_l.
\end{align}
Note that $\rho \geq 0$ and $\rho' \geq 0$ since $0\leq \Delta \leq \frac{1}{3}$.
In order to compute the sub-optimality gaps \eqref{suboptimality} we compute the following quantities for $\sigma_i\in\mathcal{A}$,
\begin{align}\label{mean1pauli}
\Tr(\rho \sigma_i) = \frac{\Delta}{d}\Tr(\sigma_1\sigma_i) = \begin{cases}
\Delta \quad \text{if}\quad  i=1.\\
0 \quad \text{otherwise.}
\end{cases}
\end{align}
\begin{align}\label{mean2pauli}
\Tr(\rho' \sigma_i) = \frac{\Delta}{d}\left( \Tr(\sigma_1\sigma_i)+2\Tr(\sigma_l\sigma_i) \right) = \begin{cases}
2\Delta \quad i=l. \\ \Delta \quad i=1 . \\ 0 \quad \text{otherwise.}
\end{cases}
\end{align}
For $\rho$ the sub-optimality \eqref{suboptimality} gaps are $\Delta_i = \Delta$ for $i\neq 1$ and $\Delta_i = 0 $ for $i=1$. Thus we can compute the regret as,
\begin{align}
R_n(\mathcal{A},\rho , \pi) = \sum_{a=1}^k \EX_{\rho,\pi}[ T_a(n) ]\Delta_a = \Delta \sum_{a\neq 1} \EX_{\rho,\pi}[ T_a(n) ]
\end{align}
Using Markov's inequality and $n = \sum_{i=a}^k \EX_{\rho,\pi}[ T_a(n) ]$,
\begin{align}\label{pauli1}
R_n(\mathcal{A},\rho , \pi ) = \Delta \EX_{\rho,\pi}\left[ n - T_1(n) \right] \geq \frac{n\Delta}{2}P_{\rho , \pi}\left( T_1(n) \leq \frac{n}{2} \right).
\end{align}

For $\rho'$ we bound the regret for the term $i=1$, using that the sub-optimality gap is $\Delta_1 = \Delta$ and Markov's inequality we have,
\begin{align}\label{paulil}
R_n(\mathcal{A},\rho' , \pi ) = \sum_{a=1}^k \EX_{\rho',\pi}[ T_a(n) ]\Delta_a \geq \Delta \EX_{\rho',\pi}[ T_1(n) ]\geq \frac{n\Delta}{2} P_{\rho' , \pi}\left( T_1(n) > \frac{n}{2} \right) .
\end{align}

Combining Equations~\eqref{pauli1},\eqref{paulil} and the Bretagnolle-Huber inequality \eqref{pinsker} we have,
\begin{align}\label{pauliregsum}
R_n(\mathcal{A},\rho , \pi ) + R_n(\mathcal{A},\rho' , \pi ) \geq \frac{n\Delta}{4}\exp\left(-D(P_{\rho , \pi} \| P_{\rho' , \pi})  \right).
\end{align}

Note that the probabilities of the rewards are just Bernoulli distributions since the observables $\sigma_i\in\mathcal{A}$ have two outcomes, +1 and -1. Using Lemma \ref{divergence} we can express

\begin{align}
D (P_{\rho , \pi} \| P_{\rho' , \pi }) = \sum_{a=1}^k \EX_{\rho ,\pi} [T_a (n ) ] D \big( P_{\rho}(\cdot | a) \big\| P_{\rho'}(\cdot | a) \big) .
\end{align}
Note that $P_{\rho}(\cdot | a) , P_{\rho'}(\cdot | a) $ are not equal only when $a=l$, so $D\big( P_{\rho}(\cdot | a) \big\| P_{\rho'}(\cdot | a) \big)  = 0$ for $a\neq l$ and,
\begin{align}\label{paulidiv1}
D (P_{\rho , \pi} ,P_{\rho' , \pi }) = \EX_{\rho ,\pi} [T_l (n ) ] D\big( P_{\rho}(\cdot | l) \big\| P_{\rho'}(\cdot | l) \big) .
\end{align}
Using Equations~\eqref{mean1pauli} and \eqref{mean2pauli} we have 
\begin{align}
P_{\rho}(1 | l) = \frac{1}{2}, \quad P_{\rho'}(1 | l) = \frac{1}{2}+\Delta.
\end{align}
Then we can compute the Kullback–Leibler divergence as,
\begin{align}\label{paulidivergence}
 D\big( P_{\rho}(\cdot | l) \big\| P_{\rho'}(\cdot | l) \big)  = \frac{1}{2}\log \frac{\frac{1}{2}}{\frac{1}{2}+\Delta}+\frac{1}{2}\log \frac{\frac{1}{2}}{\frac{1}{2}-\Delta} = \frac{1}{2}\log\frac{1}{1-4\Delta^2}.
\end{align}
Note that using the definition of $l$ $\eqref{lessplayed}$ it follows that $n = \sum_{i=1}^k E_{\rho,\pi}[T_i(n) ] \geq \sum_{i=2}^k E_{\rho,\pi}[T_i(n) ] \newline \geq (k-1)E_{\rho,\pi}[T_l(n) ]$.
Thus it holds that, 
 \begin{align}\label{pauliless}
 E_{\rho , \pi} [T_l (n ) ] \leq \frac{n}{k-1}.
 \end{align}
Combining Equations~\eqref{pauliregsum},\eqref{paulidiv1},\eqref{paulidivergence} and \eqref{pauliless} we have,
\begin{align}
R_n(\mathcal{A},\rho , \pi ) + R_n(\mathcal{A},\rho' , \pi ) \geq \frac{n\Delta}{4}\exp\left(-\frac{n}{2(k-1)}\log\frac{1}{1-4\Delta^2}  \right).
\end{align}
Finally choosing $\Delta = \frac{1}{2}\sqrt{1-e^{-\frac{k-1}{n}}}$ and using that $e^{-x}\leq (e^{-1}-1)x+1$ for $0\leq x\leq 1$ we arrive to the result,
\begin{align}
R_n(\mathcal{A},\rho , \pi ) + R_n(\mathcal{A},\rho' , \pi ) \geq \frac{e^{-1/2}\sqrt{1-e^{-1}}}{8}\sqrt{(k-1)n}.
\end{align}
The condition $n\geq 2(k-1)$ suffices to have $0 \leq \Delta\leq \frac{1}{3}$. The theorem follows using $ \frac{e^{-1/2}\sqrt{1-e^{-1}}}{8} > \frac{3}{50}$.
\end{proof}

\subsection{Pure-states environments}\label{sec:Pauliobservablesdiscrete}

The last setting that we study for the discrete multi-armed quantum bandits is for pure-states environments. Specifically we restrict the problem to one-qubit pure-state environments with actions set comprised of the Pauli observables.

\begin{theorem}\label{th:pauliobservablesdiscrete}
Let $n\in \mathbb{N}$. For any policy $\pi$ and action set $\mathcal{A}$ containing the Pauli observables for 1-qubit, there exists an environment $|\psi \rangle \! \langle \psi | \in \mathcal{S}_2^*$ such that
\[   R_n(\mathcal{A},\psi , \pi  ) \geq \frac{3}{200}\sqrt{n}. \]
\end{theorem}

\begin{proof}
Let $\Delta \in [0,1]$ and define the following the following two pure-states environments,
\begin{align}
| \psi_{\pm} \rangle := c_{\pm} \left( 1+\frac{1\pm\Delta}{\sqrt{2}} |0 \rangle +  \frac{1+\Delta}{\sqrt{2}} |1\rangle \right),
\end{align}
where $c_{\pm} = \frac{1}{\sqrt{\left( 1+\frac{1\pm\Delta}{\sqrt{2}} \right)^2 + \left(\frac{1\pm\Delta}{\sqrt{2}}  \right)^2 }}$.
Let $|0\rangle, |+\rangle , |\psi_y \rangle$ be the eigenvector of the Pauli observables $\sigma_z , \sigma_x , \sigma_y $ respectively and define $x_{\pm} = \frac{1\pm \Delta}{\sqrt{2}}$. Then, compute the following quantities
\begin{align}\label{probpure}
p^x_\pm = | \langle + | \psi_{\pm} \rangle |^2 = \frac{(1+2x_\pm )^2}{2\left( (1+x_{\pm})^2+x_{\pm}^2 \right)} , \quad
p^z_\pm = | \langle 0 | \psi_{\pm} \rangle |^2 = \frac{(1+x_{\pm})^2}{(1+x_{\pm})^2+x_{\pm}^2} , \\
p^y_\pm = | \langle \psi_y | \psi_\pm \rangle |^2 = \frac{1}{2}. \nonumber
\end{align}
Note that the expectation values for each arm on the environments  can be computed as,
\begin{align}\label{gappure}
\Tr \left( \sigma_i |\psi_\pm \rangle \! \langle \psi_\pm | \right) = 2p^i_\pm -1,
\end{align}
for $i = x,y,z$. Note that $\Tr \left( \sigma_y |\psi_\pm \rangle \! \langle \psi_\pm | \right) = 0$, 
$\Tr \left( \sigma_x |\psi_+ \rangle \! \langle \psi_+ | \right) \geq \Tr \left( \sigma_z |\psi_+ \rangle \! \langle \psi_+ | \right)  \geq $ 0 and 
$\Tr \left( \sigma_z |\psi_- \rangle \! \langle \psi_- | \right) \geq \Tr \left( \sigma_x |\psi_- \rangle \! \langle \psi_- | \right) \geq 0 $  for $\Delta \in [0,1]$. Thus, for the environment $|\psi_+ \rangle$ the optimal action is given by $\sigma_x$ and for $|\psi_- \rangle$ by $\sigma_z$. Let $\Delta^{\pm}_i$ denote the sub-optimality gap for the $i-$th action in the environment $\psi_\pm$ respectively. Using Equations~\eqref{probpure} and \eqref{gappure} we can compute the sub-optimality gaps as,

\begin{align}\label{puregaps}
\Delta^+_x = 0, \quad 
\Delta^+_z = \frac{2+\Delta}{(1+\Delta)^2 + \sqrt{2}(1+\Delta)+1}\Delta, \quad 
\Delta^+_y = \frac{(\Delta+1)(1+\sqrt{2}+\Delta )}{\Delta^2 + (2+\sqrt{2})\Delta +2 +\sqrt{2}} \\
\Delta^-_x = \frac{2-\Delta}{(1-\Delta)^2 + \sqrt{2}(1-\Delta ) + 1}\Delta,\quad 
\Delta^-_z = 0 , \quad 
\Delta^-_y = \frac{-\sqrt{2}\Delta +1+\sqrt{2}}{\Delta^2-(\sqrt{2}+2)\Delta+\sqrt{2}+2}.
\end{align}

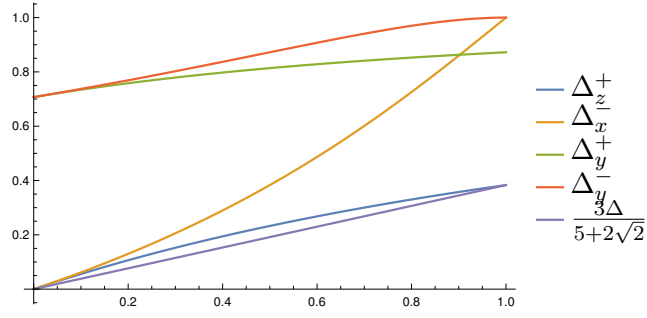
\begin{figure}[h]
\centering
\begin{overpic}[percent,width=0.5\textwidth]{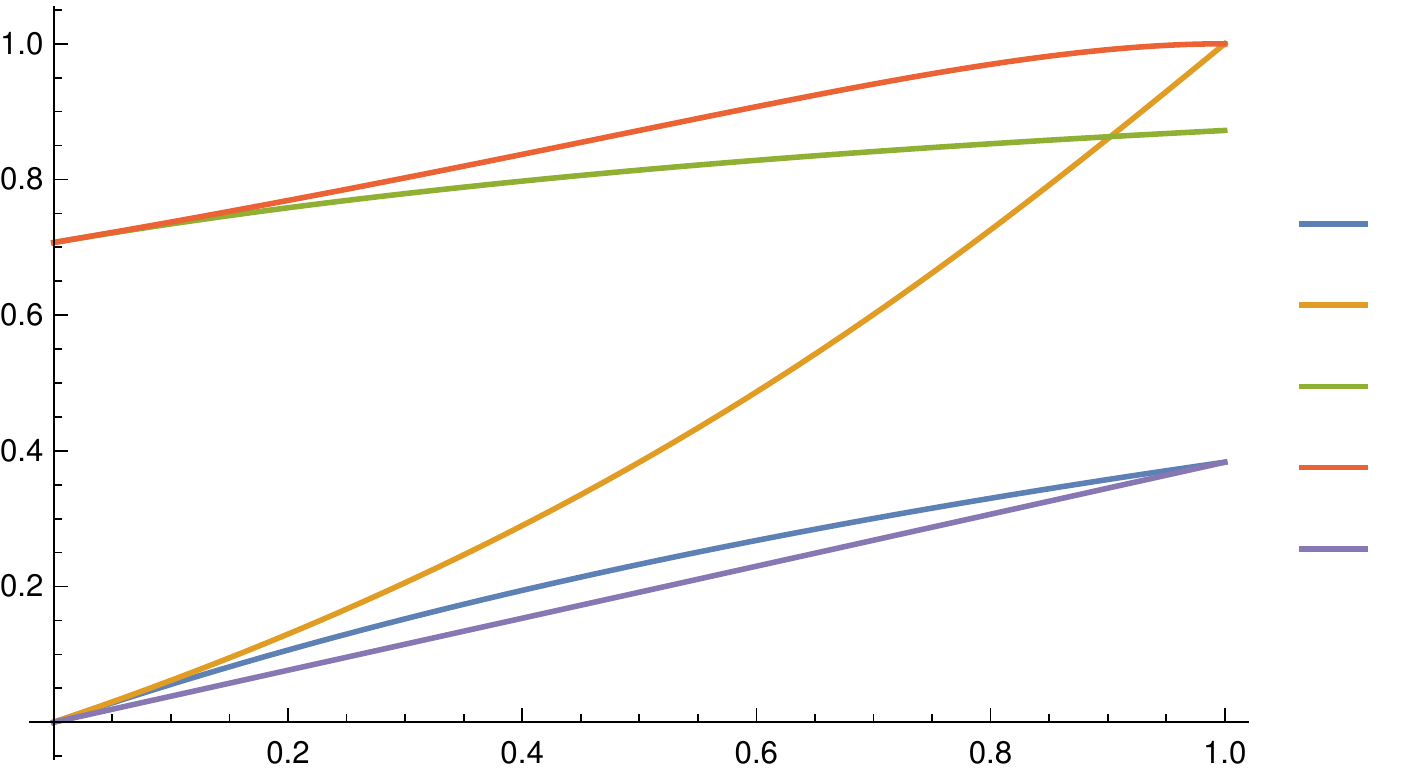}
\put(98,38){$\Delta^+_z$}
\put(98,33){$\Delta^-_x$}
\put(98,27){$\Delta^+_y$}
\put(98,21){$\Delta^-_y$}
\put(98,15){$\frac{3\Delta}{5+2\sqrt{2}}$}
\end{overpic}
\centering
\caption{Sub-optimality gaps for the environments $\psi_+,\psi_-$ and the lower bound $\frac{3}{5+2\sqrt(2)}\Delta $ in the range $\Delta \in [0,1]$.}
\end{figure}

Analysing the behaviour of $\Delta^+_z,\Delta^+_y,\Delta^-_x,\Delta^-_y$ for $\Delta \in [0,1]$ we see that all these quantities are lower bounded by $\frac{3}{5+2\sqrt(2)}\Delta$. In order to simplify the expression we will use that $\frac{3}{5+2\sqrt(2)}\Delta > \frac{3}{10}\Delta  $ and take $\frac{3}{10}\Delta$ as the lower bound.

We start analysing the regret for $\psi_+$. Using the Equations~\eqref{puregaps} for the sub-optimality gaps and the bound $\Delta^+_z,\Delta^+_y\geq \frac{3}{10}\Delta$ we have

\begin{align}
R_n(\mathcal{A},\psi_+ , \pi  ) = \sum_{i=x,y,z} \EX_{\psi_+ , \pi}(T_i(n))\Delta^+_i \geq \frac{3}{10}\Delta \sum_{i=y,z}\EX_{\psi_+ , \pi}(T_i(n)).
\end{align}

Now using $n = \sum_{i=x,y,z} \EX_{\psi_+ , \pi } (T_i (n) ) $ and applying Markov's inequality we have

\begin{align}\label{regret+}
R_n(\mathcal{A},\psi_+ , \pi  ) \geq \frac{3}{10}\Delta \EX_{\psi_+,\pi}\left( n - T_x (n)  \right) \geq \frac{3n\Delta}{20}P_{\psi_+ , \pi} \left( T_x (n) \leq \frac{n}{2} \right).
\end{align}

In order to bound the regret for the environment $\psi_-$ we bound just the first term, use Markov inequality and the bound $\Delta^-_x$ in order to obtain

\begin{align}\label{regret-}
R_n(\mathcal{A},\psi_-, \pi  ) = \sum_{i=x,y,z} \EX_{\psi_- , \pi}(T_i(n))\Delta^-_i \geq \Delta^-_x \EX_{\psi_- , \pi}(T_i(n)) \geq \frac{3n\Delta}{20} P_{\psi_- , \pi} \left( T_x (n) > \frac{n}{2} \right).
\end{align}

Combining Equations~\eqref{regret+} and \eqref{regret-}
\begin{align}
R_n(\mathcal{A},\psi_+ , \pi  ) + R_n(\mathcal{A},\psi_-, \pi  )  \geq \frac{3n\Delta}{20} \left( P_{\psi_+ , \pi} \left( T_x (n) \leq \frac{n}{2} \right) + P_{\psi_- , \pi} \left( T_x (n) > \frac{n}{2} \right) \right).
\end{align}

Applying Lemma \ref{pinsker2} together with Lemma \ref{lemma:D1/2ineq} we obtain

\begin{align}\label{discretesumreg2}
R_n(\mathcal{A},\psi_+ , \pi  ) + R_n(\mathcal{A},\psi_-, \pi  )  \geq \frac{3n\Delta}{40}\exp \left( -n D_{\frac{1}{2}} (\psi_+ \| \psi_- )  \right),
\end{align}
where $D_{\frac{1}{2}} (\psi_+ \| \psi_- ) = -\log | \langle \psi_+ | \psi_- \rangle |^2$. The overlap between the two environments can be computed as
\begin{align}
\langle \psi_+ | \psi_- \rangle = 2+\sqrt{2} - \Delta^2,
\end{align}
and we use it to give the following upper bound
\begin{align}\label{logoverlap}
-\log | \langle \psi_+ | \psi_- \rangle |= \log \frac{1}{2+\sqrt{2}-\Delta^2} \leq \frac{\Delta^2 -1 -\sqrt{2}}{2+\sqrt{2}-\Delta^2}\leq \frac{\Delta^2}{2+\sqrt{2}-\Delta^2} \leq \frac{\Delta^2}{1+\sqrt{2}},
\end{align}
where the last inequality follows from $\frac{1}{2+\sqrt{2}-\Delta^2} \leq \frac{1}{1+\sqrt{2}}$ for $\Delta \in [0,1]$.

Thus, plugging Equation~\eqref{logoverlap} into Equation~\eqref{discretesumreg2} we have
\begin{align}
R_n(\mathcal{A},\psi_+ , \pi  ) + R_n(\mathcal{A},\psi_-, \pi )  \geq \frac{3n\Delta}{40}\exp \left( -\frac{2}{1+\sqrt{2}} n\Delta^2 \right).
\end{align}
Finally, if we choose $\Delta = \frac{1}{\sqrt{n}}$,

\begin{align}
R_n(\mathcal{A},\psi_+ , \pi  ) + R_n(\mathcal{A},\psi_-, \pi  )  \geq \frac{3}{40} \exp \left( -\frac{2}{1+\sqrt{2}}  \right) \sqrt{n},
\end{align}
and the result follows using $ \frac{3}{40} \exp \left( -\frac{2}{1+\sqrt{2}}  \right) \geq \frac{3}{100}$.
\end{proof}

\section{Regret lower bounds for general bandits}\label{sec:lowergeneral}
\subsection{Rank-1 projectors for general environments}\label{sec:AllPureStates}

In the following theorem, we consider a bandit whose action set is the set of all rank-1 projections, i.e.,\ $\mc A=\mc S_d^*$ where $d\in\N$ is the dimension of the Hilbert space. When we play the arm $|\fii\>\!\<\fii|\in S_d^*$ on the environment state $\rho$, the probability of reward 1 is $\<\fii|\rho|\fii\>$ and the probability of reward 0 is $1-\<\fii|\rho|\fii\>$. 
Denoting by $\lambda_{\rm max}(\rho)$ the highest eigenvalue of any environment state $\rho$, the sub-optimality gap is $\Delta_\fii:=\lambda_{\rm max}(\rho)-\<\fii|\rho|\fii\>$ for the environment state $\rho$ upon playing the arm $|\fii\>\<\fii|$.

\begin{theorem}\label{theor:AllPureStates}
Let $n,\,d\in\N$. For any policy $\pi$ and action set of observables $\mathcal{A}$ containing all rank-1 projections, i.e.,\ $\mc A=\mc S_d^*$, there exists an environment $\rho\in\mc S_d$ such that
\begin{align}
R_n(\mc A,\rho,\pi)\geq C_{\mathcal{A}}\sqrt{n}
\end{align}
for some constant $C_{\mathcal{A}}>0$ that depends on the action set.
\end{theorem}

\begin{proof}
Let us fix a policy $\pi$ and an orthonormal basis $\{|n\>\}_{n=0}^{d-1}$ for $\C^d$ and define $|\psi\>:=d^{-1/2}(|0\>+\cdots+|d-1\>)$ and the sets
\begin{align}
\mc N_1&:=\left\{|\eta\>\!\<\eta|\in\mc S_d^*\,\middle|\,|\<0|\eta\>|^2<\frac{3}{4}+\frac{1}{4d}\right\},\\
\mc N_2&:=\left\{|\eta\>\!\<\eta|\in\mc S_d^*\,\middle|\,|\<\psi|\eta\>|^2<\frac{3}{4}+\frac{1}{4d}\right\}.
\end{align}
Let us show that these sets are disjoint. First note that $P(|0\>\!\<0|,|\psi\>\!\<\psi|)=\sqrt{1-1/d}$ where $P=\sqrt{1-F^2}$ is the purified distance where, in turn, $F$ is the fidelity. Assume that $|\eta\>\!\<\eta|\in\mc N_1$, so that, using the triangle inequality for the purified distance,
\begin{align}
\sqrt{1-\frac{1}{d}} &= P(|0\>\!\<0|,|\psi\>\!\<\psi|)\leq P(|0\>\!\<0|,|\eta\>\!\<\eta|)+P(|\eta\>\!\<\eta|,|\psi\>\!\<\psi|)\\
&< \frac{1}{2}\sqrt{1-\frac{1}{d}}+P(|\eta\>\!\<\eta|,|\psi\>\!\<\psi|)
\end{align}
where we have used the definition of $\mc N_1$ in the final inequality. Thus, $P(|\eta\>\!\<\eta|,|\psi\>\!\<\psi|)>(1/2)\sqrt{1-1/d}$ which is easily seen to imply $|\eta\>\!\<\eta|\in\mc N_2^c$. Thus, $\mc N_1\cap\mc N_2=\emptyset$.

Let us now define 
\begin{align}
\rho:= \frac{1-\Delta}{d}I+\Delta|0\>\!\<0| \, ,
\end{align}
where $\Delta\in[0,1]$ is a constant to be determined later. For this environment, the sub-optimality gap is
\begin{align}
\Delta_\eta=\Delta(1-|\<0|\eta\>|^2).
\end{align}
We may now evaluate
\begin{align}
R_n(\mc A,\rho,\pi) &\geq \Delta\int_{\mc N_1^c}(1-|\<0|\eta\>|^2)\,d\gamma_{\rho,\pi}(|\eta\>\!\<\eta|)\geq\Delta\frac{d-1}{4d}\gamma_{\rho,\pi}(\mc N_1^c).
\end{align}
Let us define another state 
\begin{align}
\rho':= \frac{1-\Delta}{d}I+\Delta|\psi\>\!\<\psi| \, ,
\end{align}
 where $|\psi\>$ is the unit vector defined earlier. Similarly as above, we find
\begin{align}
R_n(\mc A,\rho',\pi) &\geq \Delta\int_{\mc N_2^c}(1-|\<\psi|\eta\>|^2)\,d\gamma_{\rho',\pi}(|\eta\>\!\<\eta|)\geq\Delta\frac{d-1}{4d}\gamma_{\rho',\pi}(\mc N_2^c)\\
&\geq \Delta\frac{d-1}{4d}\gamma_{\rho',\pi}(\mc N_1)
\end{align}
where the final inequality follows from $\mc N_1\subseteq\mc N_2^c$; recall that $\mc N_1$ and $\mc N_2$ are disjoint. Recalling that $(1/n)\gamma_{\rho,\pi}$ and $(1/n)\gamma_{\rho',\pi}$ are probability measures and using Lemma \ref{pinsker}, we now obtain
\begin{align}\label{eq:summaepayhtalo}
R_n(\mc A,\rho,\pi)+R_n(\mc A,\rho',\pi) &\geq\Delta\frac{d-1}{4d}n\left(\frac{1}{n}\gamma_{\rho,\pi}(\mc N_1^c)+\frac{1}{n}\gamma_{\rho',\pi}(\mc N_1)\right) \\
&\geq\Delta\frac{d-1}{4d}n\exp{\left(-D\left(\frac{1}{n}\gamma_{\rho,\pi}\middle\|\frac{1}{n}\gamma_{\rho',\pi}\right)\right)}.
\end{align}

Recall that, e.g.,\ $\gamma_{\rho,\pi}=\sum_{t=1}^n\big(P_{\rho,\pi}^{(t)}(\cdot,0)+P_{\rho,\pi}^{(t)}(\cdot,1)\big)$. Defining, for all $t=1,\ldots,\,n$, the (measurable) function $f_t:(\mc A\times\{0,1\})^n\to\mc A$ through
\begin{align}
f_t\big((|\eta_s\>\!\<\eta_s|,x_s)_{s=1}^n\big)=|\eta_t\>\!\<\eta_t|,
\end{align}
we now see that $\gamma_{\rho,\pi}=\sum_{t=1}^nP_{\rho,\pi}\circ f_t^{-1}$ where we view $P_{\rho,\pi}$ as a measure. Using the joint convexity of the Kullback-Leibler divergence and the data processing inequality, we finally get
\begin{align}
D\left(\frac{1}{n}\gamma_{\rho,\pi}\middle\|\frac{1}{n}\gamma_{\rho',\pi}\right)
&\leq \frac{1}{n}\sum_{t=1}^n D(P_{\rho,\pi}\circ f_t^{-1}\|P_{\rho',\pi}\circ f_t^{-1})\\
&\leq \frac{1}{n}\sum_{t=1}^n D(P_{\rho,\pi}\|P_{\rho',\pi})=D(P_{\rho,\pi}\|P_{\rho',\pi})\\
&\leq nD(\rho\|\rho')
\end{align}
where the final inequality follows from Lemma \ref{lemma:D1/2ineq} (for $\alpha=1$). Combining this with Equation~\eqref{eq:summaepayhtalo} gives us
\begin{align}
R_n(\mc A,\rho,\pi)+R_n(\mc A,\rho',\pi)\geq\Delta\frac{d-1}{4d}n\exp{\big(-nD(\rho\|\rho')\big)}.
\end{align}

Whenever $\Delta\leq 1/2$, we may follow the end of the proof of Theorem \ref{generalower} and obtain $c>0$ such that $D(\rho\|\rho')\leq (c/2)\Delta^2$. Fixing $\Delta=1/\sqrt{n}$, we now have
\begin{align}
\max\{R_n(\mc A,\rho,\pi),R_n(\mc A,\rho',\pi)\}\geq\frac{1}{2}\big(R_n(\mc A,\rho,\pi)+R_n(\mc A,\rho',\pi)\big)\geq\frac{d-1}{8d}\sqrt{n}e^{-c/2}.
\end{align}
\end{proof}

We would like to mention that for continuous action sets there are different methods for proving minimax regret lower bounds for classical stochastic linear bandit (see \cite{banditalgorithm}[Chapter 24] or \cite{unitsphere}). While the generalization to quantum is far from trivial we expect that alternatives techniques would be able to extract a non-trivial dimensional dependence on our regret lower bound.

\section{Algorithms and regret upper bounds}\label{sec:algorithms}

In this section we are going to review some of the multi-armed stochastic bandit algorithms and see how they can be implemented in the multi-armed quantum bandit case. First we are going to review the linear stochastic bandits where the expected reward has a linear structure in terms of the actions and explain the \textsf{LinUCB} (linear upper confidence bound) algorithm. We will see that the \textsf{LinUCB} algorithm can be applied to both discrete or continuous sets of actions for multi-armed quantum bandits. For the discrete case we are going to review the \textsf{UCB} algorithm and the \textsf{Phased Elimination} algorithm. The first one does not assume correlations between the arms whereas the second one does. For different regimes of the number of actions, these algorithms offer better scaling on the regret than \textsf{LinUCB}.

First of all we are going to give an expression for the rewards of the multi-armed quantum bandit model where there is a linear part that depends on the actions and a random part that comes from some subgaussian noise. 
We introduce the notion of subgaussianity. We say that a random variable $X$ is $R$-subgaussian if for all $\mu \in \mathbb{R}$ we have,
\begin{align}\label{subgaussian}
\EX \left[ \exp (\mu X ) \right] \leq \exp ( R^2 \mu^2 / 2 ).
\end{align}
The notion of subgaussianity implies that $\EX [X] = 0$ and $\mathbb{V}[X] \leq R^2$.

The following Lemma will allow us to match the linear structure of the rewards of the multi-armed quantum bandit problem to the classical bandit problem.

\begin{lemma}\label{lem:subgaussianrewards}
Let $n,d\in\mathbb{N}$. Consider a general or discrete multi-armed quantum bandit problem $(\mathcal{A},\Gamma)$  with environment $\rho\in\Gamma$ such that, for any $O\in\mathcal{A}$, $\|O\| \leq 1$. Then if $O_t\in\mathcal{A}$ is the observable selected at round $t\in[n]$,  the rewards are of the form
\begin{align}
X_t = \boldsymbol{\theta}\cdot \mathbf{A}_t + \eta_t,
\end{align}
where $\boldsymbol{\theta}, \mathbf{A}_t\in \mathbb{R}^{d^2}$, $\Tr (\rho O_t ) = \boldsymbol{\theta}\cdot \mathbf{A}_t$ and $\eta_t$ is 1-subgaussian given $X_1,O_1,...,X_{t-1},O_{t-1}$.
\end{lemma}

\begin{proof}
Let $\left\lbrace \sigma_i \right\rbrace_{i=1}^{d^2}$ be a set of orthonormal $(\Tr\left(\sigma_i \sigma_j \right) = \delta_{ij})$ Hermitian matrices ($\sigma_i^\dagger =\sigma_i $). Then we have that for any $\rho\in \Gamma$ and action $O_a \in \mathcal{O}_d$

\begin{align}\label{dparametrization}
\rho = \sum_{i=1}^{d^2} \theta_i \sigma_i, \quad O_a = \sum_{i=1}^{d^2} A_{a,i} \sigma_i,
\end{align}
where $\theta_i = \Tr(\rho \sigma_i)$ and $A_{a,i} = \Tr (O_a \sigma_i )$. We will denote $\boldsymbol{\theta}\in\mathbb{R}^{d^2}$ the vector associated to $\rho$ with components $\theta_i$ and $\bold{A}_a\in\mathbb{R}^{d^2}$ the vector associated to $O_a$ with components $A_{a,i}$. Note that $\Tr(\rho O_a ) = \boldsymbol{\theta}\cdot \bold{A}_a$ since the set $\left\lbrace \sigma_i \right\rbrace_{i=1}^{d^2} $ is orthonormal. 

At round $t$ define $\eta_t = X_t - \Tr (\rho O_t )$ where $X_t$ is the reward and $O_t$ the observable selected by the learner. Recall that $\EX [X_t] =  \Tr (\rho O_t )$ and using the assumption $\| O_a \| \leq 1$ we have $|X_t|\leq 1$.
Thus, if he apply Hoeffding Lemma~\cite[Equation 4.16]{Hoeffding} to $\eta_t$ we have that for any $\lambda \in \mathbb{R}$,
\begin{align}
\EX [ \exp (\lambda \eta_t )]\leq \exp \left( \frac{\lambda^2}{2}\right).
\end{align}
By definition of subgaussian \eqref{subgaussian} it follows that $\eta_t$ is 1-subgaussian, and the result follows.

\end{proof}

\subsection{\textsf{LinUCB} algorithm for general multi-armed quantum bandits}\label{linucb}

In this section we will review one algorithm that has been used for classical bandits named \textsf{LinUCB} (linear upper confidence bound) or \textsf{LinRel} (linear reinforcement learning). The classical bandits are analogous to our quantum case with the main difference that instead of performing a measurement at each round we sample a reward from a set of probability distributions. In order to introduce the \textsf{LinUCB} algorithm we quickly review stochastic linear bandits.

The linear bandit model is described as follows.
Let $\boldsymbol{\theta} \in \mathbb{R}^d$ be an unknown vector and $\mathcal{A} \subset \mathbb{R}^d$ be a set of vectors that we call the action set. Then we have a learner, that during a sequence of $n$ rounds, at each round $t\in \left\lbrace 1,...,n\right\rbrace$ selects a vector $\bold{A}_t \in \mathcal{A}$ and samples a reward 
\begin{align}\label{clasreward}
X_t =  \boldsymbol{\theta} \cdot \bold{A}_t+ \eta_t
\end{align}
where $\eta_t$ is $R$-subgaussian given $\bold{A}_1 , X_1 , ... , \bold{A}_{t-1},X_{t-1}$ and we call it the \textit{gaussian noise}. The regret is defined as 
\begin{align}\label{clasregret}
R_n (\boldsymbol{\theta},\pi ,\mathcal{A} ) =\sum_{t=1}^n \max_{\bold{A}_i \in \mathcal{A}} \boldsymbol{\theta}\cdot  \bold{A}_i  - \EX_{\boldsymbol{\theta},\pi} [ \boldsymbol{\theta} \cdot \bold{A}_t ] .
\end{align}
The expectation value above is taken with respect to the probability distribution $P_{\boldsymbol{\theta},\pi}$ determined by the policy $\pi$ analogously to Equation \eqref{probdens} where the conditional probabilities $P_{\rho}(\cdot|\cdot)=P_{\boldsymbol{\theta}}(\cdot|\cdot)$ are now some subgaussian conditional distributions which depend on the specific problem we are studying. In order to analyse the regret of an algorithm it is convenient to define the \textit{pseudo-regret} as
\begin{align}\label{pseudo-regret}
\hat{R}_n (\boldsymbol{\theta},\pi ,\mathcal{A} ) =\sum_{t=1}^n \max_{\bold{A}_i \in \mathcal{A}} \boldsymbol{\theta}\cdot  \bold{A}_i  -  \boldsymbol{\theta} \cdot \bold{A}_t.
\end{align}
Note that $ \EX_{\boldsymbol{\theta},\pi}[\hat{R}_n ] = R_n$.
Now we introduce the \textsf{LinUCB} algorithm and explain how it fits to our quantum model. The first paper that studied an algorithm for linear stochastic bandits was \cite{firstlin} where they considered a finite action set. The \textsf{LinUCB} algorithm that we shall now introduce is based on \cite{lin1,lin2,lin3} where they allow an infinite action set and are based on statistical techniques in order to estimate the unknown parameter $\boldsymbol{\theta}$. The constructions and proofs that we will explain can be found in these references.

The \textsf{LinUCB} is an algorithm that minimizes the regret for linear stochastic bandits and it works constructing a confidence ellipsoid region for the unknown parameter $\boldsymbol{\theta}$ and selecting the best action on $\mathcal{A}$ that maximizes the reward in the confidence region. The confidence region is constructed via the least-square estimator,
\begin{align}
\hat{\boldsymbol{\theta}}_t = \argmin_{\boldsymbol{\theta}\in \mathbb{R}^d} \left( \sum_{s=1}^t \left( X_s -  \boldsymbol{\theta} \cdot \bold{A}_s  \right)^2 + \lambda \| \boldsymbol{ \theta} \|_2^2 \right),
\end{align}
where $X_s$ is the reward received at round $s$, $\bold{A}_s$ the action selected from $\mathcal{A}$ at round $s$ and $\lambda \geq 0$ is a regulator parameter that ensures that the function has a unique minimum. The above equation can be solved for $\boldsymbol{\theta}$, which has the following closed form:
\begin{align}\label{estimator}
\hat{\boldsymbol{\theta}}_t = V_t^{-1}\sum_{s=1}^t \bold{A}_s X_s
\end{align}
where $V_t$ is a $d\times d$ matrix,
\begin{align}
V_t = \lambda I + \sum_{s=1}^t \bold{A}_s \bold{A}_s^T .
\end{align}
Note that $V_t$ is positive definite by construction and induces the norm $\| \mathbf{x} \|^2_{V_t} = \mathbf{x}^T V_t \mathbf{x}$ for any $\mathbf{x}\in \mathbb{R}^d$. Using the estimator \eqref{estimator} we can build at each round $t$ a confidence region $\mathcal{C}_t \subset \mathbb{R}^d$. Assume that the gaussian noise $\eta_t$ for each $X_t$ is $1$-subgaussian and let $m$ be an upper bound on $\| \boldsymbol{\theta} \|_2$, $L$ an upper bound for any action, i.e for any $\boldsymbol{A} \in \mathcal{A}$, $\| \boldsymbol{A}\|_2 \leq L$, and $\delta \in (0,1)$. Then the following theorem gives us the exact form of the confidence region:

\begin{theorem}[Theorem 20.5 in \cite{banditalgorithm}]\label{teo:regretlinucb}
Let $d,n\in\mathbb{N}$, $t\in [n]$ and $\delta \in (0,1)$. Define a classical linear bandit problem with action set $\mathcal{A}\subset \mathbb{R}^d$ and hidden parameter $\boldsymbol{\theta} \in \mathbb{R}^d$ such that for any $\mathbf{A}\in\mathcal{A}$, $\| \mathbf{A} \|_2 \leq L$ and $\| \boldsymbol{\theta} \|_2 \leq m$ for some $L,m\in\mathbb{R}$. Assume that the rewards are of the form \eqref{clasreward} with 1-subgaussian noise. Then, $P_{\boldsymbol{\theta},\pi}({\rm exists }\ t\in[n]: \boldsymbol{\theta} \notin \mathcal{C}_t)\leq \delta$ with 
\begin{align}
\mathcal{C}_t = \left\lbrace \boldsymbol{\theta}^* \in \mathbb{R}^d : \| \boldsymbol{\theta}^* - \hat{\boldsymbol{\theta}}_{t-1} \|_{V_{t-1}}\leq m\sqrt{\lambda} + \sqrt{2\log\left( \frac{1}{\delta}+ d\log\left( \frac{d\lambda + nL^2}{d\lambda}\right) \right)} \right\rbrace.
\end{align}
\end{theorem}

\begin{algorithm}[H]
	\caption{\textsf{LinUCB}} 
	\label{alg:LinUCB}
	\begin{algorithmic}[1]
		\For {$t=1,2,\ldots$}
			\State $(\bold{A}_t, \tilde{\boldsymbol{\theta}}_t) = \argmax_{\bold{A}\in\mathcal{A},\boldsymbol{\theta}\in\mathcal{C}_t}\langle \boldsymbol{\theta}, \bold{A} \rangle$;
			\State Select $\bold{A}_t$ and observe reward $X_t$;
			\State Update $\mathcal{C}_t$;
			\EndFor
	
	\end{algorithmic} 
\end{algorithm}

The \textsf{LinUCB} algorithm works as follows: at each round $t\in \left\lbrace 1,...,n \right\rbrace$ it constructs the confidence interval $\mathcal{C}_t$ and chooses and estimator $\tilde{\boldsymbol{\theta}}_t$ and action that maximizes the reward, i.e $ (\bold{A}_t, \tilde{\boldsymbol{\theta}}_t) = \argmax_{\bold{A}\in\mathcal{A},\boldsymbol{\theta}\in\mathcal{C}_t}\langle \boldsymbol{\theta}, \bold{A} \rangle$, and selects $\bold{A}_t$ in order to sample the reward. The pseudo-code can be found in Algorithm \ref{alg:LinUCB}.

\begin{theorem}[Theorem 19.2 in \cite{banditalgorithm}]\label{thregret}
Under the assumptions of Theorem \ref{teo:regretlinucb} for $\boldsymbol{\theta}\in\mathbb{R}^d,\mathcal{A}\subseteq \mathbb{R}^d$ the pseudo-regret \eqref{pseudo-regret} of \textsf{LinUCB} satisfies
\begin{align}
\hat{R}_n (\boldsymbol{\theta},\pi,\mathcal{A} ) \leq \sqrt{8dn\beta_n\log \left( \frac{d\lambda + nL^2}{d\lambda} \right)}
\end{align}
where
\begin{align}
\sqrt{\beta_n} = \sqrt{\lambda}m + \sqrt{2\log\left( \frac{1}{\delta}\right) +d\log\left( \frac{d\lambda + nL^2}{d\lambda}. \right)}
\end{align}
In particular the regret \eqref{clasregret} of \textsf{LinUCB} with $\delta = 1/n$ is bounded by
\begin{align}
R_n (\boldsymbol{\theta},\pi,\mathcal{A} )  \leq Cd\sqrt{n}\log(n),
\end{align}
where $n\in\mathbb{N}$ and $C>0$ is a suitably large universal constant. 
\end{theorem}

The following theorem shows how the \textsf{LinUCB} can be applied to the multiarmed quantum bandit problem and the scaling of the regret.

\begin{theorem}\label{teo:linucbregretmaqb}
Let $d,n\in\mathbb{N}$, and consider a general or discrete multi-armed quantum bandit with action set $\mathcal{A}$ such that, for all $O\in \mathcal{A}$, $\| O \| \leq 1$. Then the \textsf{LinUCB} algorithm, associated to policy $\pi$, can be applied to the general or discrete multi-armed bandit problem and for some universal constant $C>0$ and for any $\rho\in\mathcal{S}_d$ the regret is bounded by
\begin{align}
R_n (\mathcal{A},\rho ,\pi  ) \leq C d^2\sqrt{n} log(n).
\end{align}
\end{theorem}

\begin{proof}
Considering a set $\lbrace \sigma_i \rbrace_{i=1}^{d^2}$ of orthonormal Hermitian matrices, we can parametrize $\rho\in\mathcal{S}_d$ and any $O_a\in\mathcal{A}$ as $\rho = \sum_{i=1}^{d^2} \theta_i \sigma_i$, $ O_a = \sum_{i=1}^{d^2} A_{a,i} \sigma_i$ where $\theta_i = \Tr(\rho\sigma_i)$ and $A_{a,i} = \Tr (O_a\sigma_i)$. We denote $\boldsymbol{\theta},\bold{A}_a\in\mathbb{R}^{d^2}$ as the vectors of components $\theta_i$ and $A_{a,i}$ respectively. If $O_t$ is the action selected at time step $t\in[n]$ we can use Lemma \ref{lemma:D1/2ineq} and express the rewards of a multi-armed quantum bandit as 

\begin{align}
X_t = \boldsymbol{\theta}\cdot \mathbf{A}_t + \eta_t,
\end{align}
where $\eta_t$ is 1-subgaussian. Thus, this rewards is of the form \eqref{clasreward}. Using that $\Tr(\rho O_a ) = \boldsymbol{\theta}\cdot \mathbf{A}_a$ and $\EX_{\rho,\pi} ( \eta_t ) = 0,$ we can express the regret for the multi-armed quantum bandit as
\begin{align}
R_n ( \mathcal{A},\rho,\pi ) = \sum_{t=1}^n \max_{\bold{A}_i}\boldsymbol{\theta}\cdot \bold{A}_i - \EX_{\rho,\pi}[\boldsymbol{\theta}\cdot\bold{A_t}],
\end{align}
which is of the form of the linear stochastic bandit \eqref{clasregret}. With these observations we see how the multi-armed quantum bandit is mapped to the linear bandit and the \textsf{LinUCB} algorithm can be used. In order to bound the regret we need to check the conditions of Theorem \ref{teo:regretlinucb}. It remains to check an upper bound for $\| \boldsymbol{\theta} \|_2$. Using that the set  $\lbrace \sigma_i \rbrace_{i=1}^{d^2}$ is orthonormal we have that $\Tr (\sigma^2_i) = 1$, so $\| \sigma \| \leq 1$. Thus,

\begin{align}
\| \boldsymbol{\theta } \|_2 = \sqrt{\sum_{i=1}^{d^2} \Tr(\rho\sigma_i)^2} \leq d.
\end{align} 

Inserting the above bound into the analysis of Theorem \ref{thregret} and taking into account that $\boldsymbol{\theta},\mathbf{A}_a\in\mathbb{R}^{d^2}$ the result follows.

\end{proof}

\subsection{\textsf{UCB} and \textsf{Phased Elimination} algorithms for discrete multi-armed quantum bandits}\label{sec:UCB}

The previous algorithm can be applied to both discrete and general multi-armed quantum bandits. For the discrete case we can apply the upper confidence bound (\textsf{UCB}) algorithm that was first introduced in \cite{firstUCB}. We will see that \textsf{UCB} can offer a tighter bound of the regret under certain condition on the number of actions and the dimension of the Hilbert space. First we are going to review the \textsf{UCB} algorithm for multi-armed stochastic bandits.

A \textit{multi-armed stochastic bandit} is defined with a set of probability distributions (environment) $\nu = (P_a : a\in \mathcal{A}) $ where $\mathcal{A}$ is the set of actions of cardinality $k = |\mathcal{A}|$ and we denote by $\mu_i$ for $i\in[k]$ the mean of each action. At each round $t\in [n]$ a learner selects an action $A_t\in\mathcal{A}$ and samples a reward $X_t \sim P_{A_t}$. The regret is defined as

\begin{align}
R_n ( \nu , \pi ) = \sum_{t=1}^n \max_{i\in[k]}\mu_i - \EX[X_t] ,
\end{align}
where $\pi$ is the policy and for each action $a\in[k]$ the sub-optimality gap is defined as $\Delta_a = \max_{i\in[k]}\mu_i - \mu_a$. Note that this model is more general than the linear bandits where the rewards where of the form \eqref{clasreward}. Moreover, the discrete multi-armed quantum bandit falls trivially in this class of bandits.

The basic idea of the \textsf{UCB} algorithm is to overestimate the mean of each action with high probability. The version of \textsf{UCB} that we consider assumes that the rewards are 1-subgaussian and it is based on the following bound of the empirical mean. If $(Z_t)_{t=1}^n$ is a sequence of independent 1-subgaussian random variables with mean $\mu$ and empirical mean $\hat{\mu}$, then for $\delta\in (0,1)$,

\begin{align}
\Pr \left( \mu \geq \hat{\mu} + \sqrt{\frac{2\log (1/\delta )}{n}} \right) \leq \delta .
\end{align}

Let $n$ be the number of rounds and $(X_t)_{t=1}^n$ be the sequence of rewards for a multi-armed stochastic bandit. For each action indexed by $a\in [k]$ we define its empirical mean as
\begin{align}
\hat{\mu}_a (n) = \frac{1}{T_a(n)}\sum_{t=1}^n X_t \mathbb{I}\lbrace A_t = a \rbrace,
\end{align}
where $A_t$ is the index of the action picked at round $t$ and $T_a(n) = \sum_{t=1}^n\mathbb{I}\lbrace A_t = a \rbrace $ the number of times that we have picked the action indexed by $a\in[k]$. Then at each round $t\in[n]$ for each action $a\in[k]$ the \textsf{UCB} algorithm assigns a possible mean $\tilde{\mu}(t)_a$ as

\begin{align}
\tilde{\mu}(t)_a = \begin{cases}
\infty \quad \text{if}\quad  T_a(t-1)=0.\\
 \hat{\mu}_a(t-1) + \sqrt{\frac{2\log(1/\delta)}{T_a(t-1)}}\quad \text{otherwise}
 \end{cases}
\end{align}
and selects the action $A_t = {\rm argmax}_{a\in[k]} \tilde{\mu}(t)_a  $. Note that the assignment of $\infty$ to the actions that have not been played ensures that each is action is played at least once during the first $k$ rounds. The pseudocode of the \textsf{UCB} algorithm can be found in Algorithm \ref{alg:UCB}.

\begin{algorithm}[H]
	\caption{UCB} 
	\label{alg:UCB}
	\begin{algorithmic}[1]
		
		\For {$t=1,2,\ldots,n$}
			\State Choose action $A_t = \argmax_{i\in [k]} \begin{cases}
\infty \quad \text{if}\quad  T_i(t-1)=0.\\
 \hat{\mu}_i(t-1) + \sqrt{\frac{2\log(1/\delta)}{T_i(t-1)}}\quad \text{otherwise.}
\end{cases}$
			\State Observe reward $X_t$ and update confidence bounds;
			
			\EndFor
	
	\end{algorithmic} 
\end{algorithm}

The following result on the scaling of the regret of \textsf{UCB} applied to discrete multi-armed quantum bandits follows from Lemma \ref{lem:subgaussianrewards} and the analysis of the \textsf{UCB} regret for stochastic bandits from Theorem 7.2 in \cite{banditalgorithm}. Recall that \textsf{UCB} can be applied to the multi-armed quantum bandit problem since they fall in the class of multi-armed stochastic bandits.
\begin{theorem}\label{th:UCB}
Let $n\in\mathbb{N}$ and consider discrete multi-armed quantum bandit with action set $\mathcal{A}$ of cardinality $k = | \mathcal{A}|$ such that for all $O\in\mathcal{A}$, $\| O \| \leq 1$. Then, the \textsf{UCB} algorithm associated to policy $\pi$, can be applied to the discrete multi-armed bandit problem with and for $\delta = \frac{1}{n^2}$ and any $\rho\in\mathcal{S}_d$ the regret can be bounded by
\begin{align}
R_n(\mathcal{A},\rho , \pi )  \leq 8\sqrt{nk\log(n)} + \sum_{a=1}^k \Delta_a .
\end{align}
\end{theorem}

Now we are going to compare the scaling of the \textsf{UCB} regret with the \textsf{LinUCB} regret for the discrete multi-armed quantum bandits. Let $k$ be the number of observables and $d$ the dimension of the Hilbert space of a discrete multi-armed bandit problem. The above Theorem shows that if $k\leq d^4$ the regret scaling of \textsf{UCB} is better than \textsf{LinUCB} (Theorem \eqref{teo:linucbregretmaqb}). Note that if we apply \textsf{UCB} to the discrete multi-armed quantum bandits the algorithm does not take advantage of the linear structure of the rewards.

We mention that there is a variant of \textsf{LinUCB} (see \cite{banditalgorithm}[Chapter 22]) that assumes a finite number of actions and can be applied to the discrete multi-armed quantum bandits. This variant is called \textsf{Phased Elimination} and it is stated in Algorithm \ref{alg:phased}. This algorithm requires the assumptions of Theorem \ref{teo:regretlinucb} and the action set $\mathcal{A}\subseteq \mathbb{R}^d$ to be finite. The analysis of {Theorem 22.1} in \cite{banditalgorithm} along with the observations of the proof of Theorem \ref{teo:linucbregretmaqb} lead immediately to the following result:
\begin{theorem}
Let $d,n\in\mathbb{N}$, and consider a discrete multi-armed quantum bandit problem with action set $\mathcal{A}$ of cardinality $k = |\mathcal{A} |$ such that for all $O \in \mathcal{A}$, $\| O \| \leq 1$. Then, the \textsf{Phased Elimination} algorithm associated to policy $\pi$, with $\delta= \frac{1}{n}$, can be applied to the discrete multi-armed bandit problem and for some universal constant $C>0$ and for all $\rho\in\mathcal{S}_d$ the regret is bounded by
\begin{align}
 R_n (\mathcal{A},\rho,\pi ) \leq Cd\sqrt{n\log(nk)}.
\end{align} 

\end{theorem}

Note that the scaling is better than \textsf{LinUCB} as long as the number of observables $k$ is not exponential in the dimension of the Hilbert space. 

\begin{algorithm}[H]
	\caption{\textsf{Phased Elimination}} 
	\label{alg:phased}
	\begin{algorithmic}[1]
		\State Set $l=1$ and $\mathcal{A}_1 = \mathcal{A}$.
		\State Let $t_l = t$ be the current time step $t$ and find $\pi_l : \mathcal{A}_l \rightarrow [0,1]$ with $\mathtt{support}(\pi_l) \leq \frac{d(d+1)}{2}$ that maximises 
		\[ \log\det \left( \sum_{\bold{a}\in\mathcal{A}} \pi(\bold{a}) \bold{a}\bold{a}^T \right) \; \;  \text{subject to } \sum_{\bold{a}\in\mathcal{A}_l} \pi(\bold{a}) = 1  \]
		\State Let $\epsilon_l = \frac{1}{2^l}$ and $T_l(\bold{a}) = \left\lceil \frac{2d\pi_l (\bold{a})}{\epsilon^2_l} \log \left( \frac{kl(l+1)}{\delta}\right) \right\rceil $ and $T_l = \sum_{\bold{a}\in\mathcal{A}_l} T_l (\bold{a}) $
		\State Choose each action $\bold{a}\in\mathcal{A}_l$ exactly $T_l (\bold{a})$ times
		\State Calculate the empirical estimate:
		\[ \hat{\boldsymbol{\theta}}_l = V^{-1}_l \sum_{t=t_l}^{t_l+T_l} \bold{A}_t X_t \;\;\text{with }  V_l = \sum_{\bold{a}	\in\mathcal{A}_l} T_l(a)\bold{aa}^T  \]
		\State Eliminate low rewarding arms:
			\[  \mathcal{A}_{l+1} = \lbrace \bold{a}\in\mathcal{A}_l : \max_{\bold{b}\in\mathcal{A}_l} \; \hat{\boldsymbol{\theta}}_l \cdot (\bold{b}-\bold{a}) \leq 2\epsilon_l \rbrace \]
		\State $l=l+1$ and \textbf{Goto Step 1}
	
	\end{algorithmic} 
\end{algorithm}

\subsection{Tomography algorithm for pure-states environments}\label{purestatealgorithm}
The last algorithm that we study is for the general model of multi-armed quantum bandits with pure-states environments. Recall that in this setting the regret can be given as,

\begin{align}
R_n(\mathcal{A}, \psi, \pi) := \frac{1}{4} \sum_ {t=1}^n \EX_{\rho , \pi} \big\| \rho - \Pi_{A_t} \big\|^2_1 \, ,
\end{align}
where $\rho = |\psi \rangle \! \langle\psi | $ is the unknown pure quantum state and $\Pi_{A_t}$ is the rank-1 projector from our action set $\mathcal{A}$ that contains all rank-1 projectors. 

The algorithm that we propose is based on the projected least squares (PLS) method given in \cite{leastsquare}. Now, we briefly review this method and explain how to apply it for the regret analysis. Let $\mathcal{H}_d$ be a $d-$dimensional Hilbert space, $\rho\in\mathcal{D}(\mathcal{H}_d )$ be an unknown quantum state that we want to estimate and $\left\lbrace M_i\right\rbrace_{i=1}^m $ the elements of a POVM that we will use to measure $\rho$. Suppose that we perform measurements over $n$ rounds and $n_i$ is the number of times that we have observed the outcome $i$ associated to the element $M_i$ from the POVM. Then knowing that the probabilities associated to each outcome $i$ are given by the Born's rule $\Tr(\rho M_i)$, the least square estimator for $\rho $ is defined as,
\begin{align}\label{leastsquareesti}
L_n = \argmin_{X\in \mathcal{O}_d} \sum_{i=1}^m \left( \frac{n_i}{n} - \Tr \left( M_i X \right)\right)^2.
\end{align}
The estimator $L_n$ is not guaranteed to be a physical state, so the next step is to project it to the physical state space as follows,
\begin{align}\label{projection}
\hat{\rho}_n = \argmin_{\tilde{\rho}\in \mathcal{S}_d} \| L_n - \tilde{\rho} \|_2,
\end{align}
where $\| \cdot· \|_2$ denotes de Frobenius norm. In \cite{leastsquare} the above estimator is studied for structured POVMs, Pauli observables and Pauli basis measurements. For our problem we will use Pauli observables. Fix $d=2^k$, and let $ \left\lbrace \sigma_i\right\rbrace_{i=1}^{d^2}$ be the set of all possible tensor products of the $2\times 2$ Pauli matrices. Note that they form a basis of the form \eqref{dparametrization}. For Pauli observables the least square estimator \eqref{leastsquareesti} has the following form
\begin{align}
L_n = \frac{1}{d}\sum_{i=1}^{d^2} \left( \frac{n_i^+ - n_i^-}{n/d^2} \right)\sigma_i,
\end{align}
where $n_i^{\pm}$ are the empirical frequencies associated to the 2-outcomes POVM $\Pi_i^{\pm} = \frac{1}{2}(I\pm\sigma_i)$ and each observable has been measured $n/d^2$ times. The next step is to take the projection \eqref{projection}. The convergence analysis in \cite{leastsquare} shows the following theorem.

\begin{theorem}[Theorem 1 in \cite{leastsquare}]\label{convergencePLS}
Let $\rho$ be a quantum state in a $d$-dimensional Hilbert space and fix a number of samples $n\in\mathbb{N}$. Then, using Pauli measurements, the PLS estimator $\hat{\rho}_n$ \eqref{projection} obeys,
\[ \Pr \left( \| \hat{\rho}_n - \rho   \|_1 \leq  \epsilon  \right) \geq 1 - e^{-\frac{n\epsilon^2}{43d^2 r^2}} \quad \mathrm{for} \quad \epsilon \in [0,1], \]
where $r = \min \left\lbrace \mathrm{rank}(\rho ) , \mathrm{rank} ( \hat{\rho}_n ) \right\rbrace $.
\end{theorem}

After reviewing the PLS estimation method we explain how we use it to construct an algorithm for our pure state problem and how to analyse the regret. Recall that we have considered an action set that contains all rank-1 projectors. For that reason we consider the algorithm for one qubit states ($d=2$) since Pauli observables for one qubit can be measured using rank-1 projectors. The pseudo-code for the algorithm can be found in Algorithm 2. If $n$ is the number of rounds, fix the error $\epsilon^2 = \frac{172\log(n)}{\sqrt{n}}$ and assume $\frac{172\log(n)}{\sqrt{n}} \leq 1$ in order to have $\epsilon\in [0,1]$. The condition $\frac{172\log(n)}{\sqrt{n}} \leq 1$ is achieved by $n \geq 8\cdot 10^6$. Then the algorithm (policy $\pi$) works as follows:
\begin{itemize}
\item During the first $\sqrt{n}$ rounds we perform the Pauli observables measurements, and calculate the estimator $\hat{\rho}_{\sqrt{n}}$ \eqref{projection} for $\rho$. Then using Theorem \ref{convergencePLS} we have,
\begin{align}\label{probesti}
P_{\rho,\pi}\left( \| \hat{\rho}_{\sqrt{n}} - \rho   \|_1 \leq  \epsilon  \right) \geq 1 - \frac{1}{n},
\end{align}
where we have used $r = 1$ since $\mathrm{rank}(\rho ) = 1$ and $\epsilon^2 = \frac{172\log(n)}{\sqrt{n}}$. In order to ensure that the estimator is pure we project the estimator into the rank-1 subspace in a $\epsilon-$ball as follows,
\begin{align}\label{rank1estimator}
\hat{\rho} = \argmin_{\rho \in \mathcal{S}^*_1} \|  \hat{\rho}_{\sqrt{n}} - \rho  \|_1 \quad \text{such that} \quad \|  \hat{\rho}_{\sqrt{n}} - \rho   \|_1\leq \epsilon.
\end{align}
Recall that for the above equation there is at least one solution with probability greater than $1-1/n$ since $\rho$ is rank-1 by assumption. Note that using \eqref{probesti} we have,
\begin{align}\label{probesti2}
P_{\rho,\pi}\left( \| \hat{\rho} - \rho   \|_1 \leq  2\epsilon  \right) \geq 1 - \frac{1}{n}.
\end{align}
\item For the remaining rounds we perform the measurement using $\hat{\rho}$ as the rank-1 projector from our action set.
\end{itemize}

\begin{algorithm}[H]
	\caption{Bandit PLS} 
	\label{alg:pure}
	\begin{algorithmic}[1]
		\For {$i=1,2,\ldots,  d^2  $}
			\For {$j=1,2,\ldots , \lceil \sqrt{n}\rceil/d^2$}
				\State Measure $\rho$ with $\sigma_i$;
				\State Update outcomes $n_i^{\pm}$;
				\EndFor
			\EndFor
		\State Compute $L_n = \frac{1}{d}\sum_{i=1}^{d^2} \left( \frac{n_i^+ - n_i^-}{n/d^2} \right)\sigma_i$;
		\State Compute 	$\hat{\rho}_{\sqrt{n}} = \argmin_{\tilde{\rho}\in \mathcal{S}_d} \| L_n - \tilde{\rho} \|_2$;
		\State Fix $ \epsilon^2 = \frac{43d^2 \log(n)}{\sqrt{n}} $
		\State Compute $\hat{\rho} = \argmin_{\rho \in \mathcal{S}^*_d} \|  \hat{\rho}_{\sqrt{n}} - \rho   \|_1 \quad \text{such that} \quad \|  \hat{\rho}_{\sqrt{n}} - \rho   \|_1\leq \epsilon$
		\For {$t= \lceil \sqrt{n}   \rceil+1,...,n$}
			\State Measure $\rho$ with $\hat{\rho}$;
		\EndFor	
	\end{algorithmic} 
\end{algorithm}

\begin{theorem}\label{th:tomographypurestates}
Let $\ket{\psi}$ be some unknown one-qubit pure state. Fix $\pi$ to be the policy associated to Algorithm \ref{alg:pure} with $d=2$ applied to the unknown state $\ket{\psi}$. Let $n\in\mathbb{N}$ denote the number of rounds, assume $\frac{172\log(n)}{\sqrt{n}} \leq 1$ $\left( n\geq 8\cdot 10^6\right)$ and let $\mathcal{A}$ be a set containing all one-qubit rank-1 projectors. Then the regret can be bounded as,
\begin{align}
 R_n (\mathcal{A},\psi,\pi ) = O\left( \sqrt{n}\log (n)\right).  
\end{align} 
\end{theorem}

\begin{proof}
First we will bound the random regret defined as,
\begin{align}
R^*_n(\mathcal{A}, \psi  , \pi) := \frac{1}{4} \sum_ {t=1}^n \big\| |\psi \rangle \! \langle\psi |  - \Pi_{A_t} \big\|^2_1 \, ,
\end{align}
and then we will take the expectation value. During the first $\sqrt{n}$ rounds the algorithm performs the Pauli measurements using rank-1 projectors and builds the estimator $\hat{\rho}$ using Equation~$\eqref{rank1estimator}$. Using Equation~\eqref{probesti2} we have the probabilistic error bound
\begin{align}\label{probreg}
P_{\psi,\pi}\left( \| |\psi \rangle \! \langle\psi | - \hat{\rho}  \|_1 \leq 2\epsilon \right)\geq 1-\frac{1}{n},
\end{align}
where $\epsilon^2 = \frac{172\log(n)}{\sqrt{n}}$. Then, with probability greater than $1-\frac{1}{n}$, the random regret can be bounded as,
\begin{align}\label{Gbound}
R^*_n (\mathcal{A},\psi,\pi )\leq  \sqrt{n} + (n-\sqrt{n}) \frac{172\log(n)}{\sqrt{n}} = O(\sqrt{n}\log(n)),
\end{align}
where the first term comes from the trivial bound $ \frac12 \| |\psi \rangle \! \langle\psi | - \Pi_t \|_1 \leq 1 $ for the first $\sqrt{n}$ rounds and the second term from the bound of Equation~\eqref{probreg} for the remaining rounds where $\Pi_t$ is the estimator \eqref{rank1estimator}. Let $C$ be the hidden constant in $O(\sqrt{n}\log(n))$, and $G$ the probabilistic event where $R^*_n (\mathcal{A},\psi,\pi )\leq C\sqrt{n}\log(n)$ and $G^C$ its complement. Note that $P_{\psi,\pi}(G^C)\leq 1/n$. Then we can calculate the regret as,
\begin{align}
R_n (\mathcal{A},\psi,\pi ) = \EX_{\psi,\pi} [  R^*_n (\mathcal{A},\psi,\pi )] = \EX_{\psi,\pi} [\mathbb{I}\left\lbrace G \right\rbrace  R^*_n (\mathcal{A},\psi,\pi )] + \EX_{\psi,\pi} [ \mathbb{I}\left\lbrace G^C \right\rbrace  R^*_n (\mathcal{A},\psi,\pi ) ]. 
\end{align}
For the first term we use the bound \eqref{Gbound} given by the event G. For the second term we use the trivial bound $R_n (\mathcal{A},\psi,\pi )\leq n$ combined with $P_{\psi,\pi}(G_i^C ) \leq  1/n$. Thus,
\begin{align}
R_n (\mathcal{A},\psi,\pi ) \leq  C\sqrt{n}\log(n)  + nP_{\psi,\pi}(G_i^C) \leq C\sqrt{n}\log(n)  +1 = O ( \sqrt{n}\log(n)   ).
\end{align}
\end{proof}

\section{Conclusions}
\label{sec:conclusion}

We have proposed a new quantum learning framework generalizing multi-armed stochastic bandits to the quantum setting. We have given fundamental bounds on the scaling of the regret for different environments and action sets that show the difficulty of the associated learning tasks. For the case of mixed-state environments and discrete bandits our upper and lower bounds match. For the case of general bandits with action set containing all rank-1 projectors our upper and lower bound match in terms of the number of rounds $n$ but there is a gap in the dimension of the system of $d^2$. However, in the case of pure-state environments, we are not yet able to show lower bounds that scale with the number of arms or the dimension of the system and we leave this as an open question. 
We showed that an algorithm based in the explore-then-commit strategy and the PLS estimator achieves a regret upper bound $R_n = O(\sqrt{n}\log (n) )$ which scales as good as \textsf{LinUCB} in this setting.  For the qubit case this problem can be mapped to stochastic linear bandits where the actions and the environments are vectors that lie in the unit sphere $S^{d-1} = \lbrace x\in\mathbb{R}^d: \|x\|_2^2 = 1 \rbrace$. In this setting, classical techniques fail to show a lower bound $R_n = \Omega ( \sqrt{n} )$, since although the unit sphere action set has been considered  (see Chapter 24 in \cite{banditalgorithm}) the lower bound is proven for environments where the 2-norm of the vector scales as $1/{\sqrt{n}}$. The problem can also be connected to phase retrieval bandits \cite{phaseretrieval} but we encounter the same issue with the lower bounds. In \cite{unitsphere} they showed a lower bound $R_n = \Omega (\sqrt{n} )$ for actions in the unit sphere using a randomization technique that consists on computing the expectation of the regret over a distribution of possible environments. This result relies heavily in some properties of the linear least square estimator that do not hold if one consider a distribution of environments over the unit sphere. Thus, we identify that there is also a gap in the classical bandit literature if one considers a stochastic linear bandit with actions and environments in the unit sphere.

\paragraph*{Acknowledgements:} 

We thank Vincent Y.~F.~Tan for inspiring this research direction by giving a guest lecture in MT's information theory module. This research is supported by the National Research Foundation, Prime Minister's Office, Singapore and the Ministry of Education, Singapore under the Research Centres of Excellence programme and the Quantum Engineering Programme grant NRF2021-QEP2-02-P05. MT is also supported in part by NUS startup grants (R-263-000-E32-133 and R-263-000-E32-731).

\bibliographystyle{quantum}
\bibliography{quantumbiblio}

\end{document}